\documentclass[11pt]{article}

\usepackage{amsmath,amsthm,amssymb,mathtools}  
\usepackage{fullpage}
\usepackage{authblk}
\usepackage{natbib}
\usepackage{bold-extra}
\usepackage{xcolor}
\usepackage[parfill]{parskip}    
\usepackage{titlesec}
\titlespacing*{\paragraph}{0pt}{1ex plus 1ex minus .2ex}{1em}
\usepackage{graphicx}
\usepackage{epstopdf}
\usepackage{seqsplit}
\usepackage{url}
\usepackage{enumerate}
\usepackage{indentfirst}
\usepackage{xcolor}
\usepackage{booktabs}
\usepackage{enumerate}

\usepackage{csquotes}
\usepackage{hyperref}

\newcommand\blfootnote[1]{%
  \begingroup
  \renewcommand\thefootnote{}\footnote{#1}%
  \addtocounter{footnote}{-1}%
  \endgroup
}

\counterwithout{equation}{section}

\DeclareMathOperator*{\argmax}{argmax}
\DeclareMathOperator*{\argmin}{argmin}

\newtheorem{thm}{Theorem}
\newtheorem{prop}{Proposition}
\newtheorem{lem}{Lemma} 
\newtheorem{cor}{Corollary}

\theoremstyle{definition}

\title{Phylogenetic Diversity Rankings in the Face of Extinctions: the Robustness of the Fair Proportion Index}
\author[1,$^\ast$]{Mareike Fischer}
\author[2]{Andrew Francis}
\author[3]{Kristina Wicke}
\affil[1]{Institute of Mathematics and Computer Science, University of Greifswald, Greifswald, Germany}
\affil[2]{Centre for Research in Mathematics and Data Science, Western Sydney University, Sydney, Australia}
\affil[3]{Department of Mathematics, The Ohio State University, Columbus (OH), USA}
\date{}                                           

\begin{document}
\maketitle

\begin{abstract}
Planning for the protection of species often involves difficult choices about which species to prioritize, given constrained resources.
One way of prioritizing species is to consider their ``evolutionary distinctiveness'', i.e. their relative evolutionary isolation on a phylogenetic tree. Several evolutionary isolation metrics or phylogenetic diversity indices have been introduced in the literature, among them the so-called Fair Proportion index (also known as the ``evolutionary distinctiveness'' score). This index apportions the total diversity of a tree among all leaves, thereby providing a simple prioritization criterion for conservation. 

Here, we focus on the prioritization order obtained from the Fair Proportion index and analyze the effects of species extinction on this ranking. More precisely, we analyze the extent to which the ranking order may change when some species go extinct and the Fair Proportion index is re-computed for the remaining taxa. We show that for each phylogenetic tree, there are edge lengths such that the extinction of one leaf per cherry completely reverses the ranking. Moreover, we show that even if only the lowest ranked species goes extinct, the ranking order may drastically change. 
We end by analyzing the effects of these two extinction scenarios (extinction of the lowest ranked species and extinction of one leaf per cherry) for a collection of empirical and simulated trees. In both cases, we can observe significant changes in the prioritization orders, highlighting the empirical relevance of our theoretical findings.
\end{abstract}
{\it Keywords:} Biodiversity conservation, Fair Proportion index, phylogenetic diversity, species prioritization

\blfootnote{$^\ast$Corresponding author\\ \textit{Email address:} \texttt{email@mareikefischer.de} (Mareike Fischer)}

\section{Introduction}
Evolutionary isolation measures or phylogenetic diversity indices have become an increasingly popular tool to prioritize species for conservation (e.g. \citet{Vanewright1991,Redding2006, Redding2008, Redding2014, Isaac2007, Vellend2011}). These indices assess the importance of species for overall biodiversity based on their placement in an underlying phylogenetic tree and can thus, next to other criteria such as threat status, serve as a prioritization tool in conservation planning.  For instance, a species that is only distantly related to others in its family may be prioritized for preservation over one that is more closely related, in order to preserve more breadth in the biodiversity.

One simple index that has been introduced in this regard is the ``Fair Proportion (FP) index'', also known as the ``evolutionary distinctiveness'' (ED) score (\citet{Redding2003,Isaac2007}). The FP index apportions the total diversity of a tree (measured as the sum of edge lengths of the tree) among all leaves by distributing each edge length equally among descending leaves. It is employed in the so-called ``EDGE of Existence'' project established by the Zoological Society of London, a conservation initiative focusing specifically at threatened species that represent a high amount of unique evolutionary history (\citet{Isaac2007}; see also \url{https://www.edgeofexistence.org/}).

In this paper, we focus on the ranking order obtained from the FP index and analyze its robustness to species extinction. More precisely, we consider the following scenario: Suppose you use the FP index to rank species for conservation, i.e. to allocate resources to protect \emph{some} but \emph{not all} species under consideration. Now suppose one or more species go extinct, for example because they did not receive any conservation attention. This will result in a change in the underlying phylogenetic tree, because some species are now extinct.  
If you now re-compute the FP index on this revised tree, how confident can you be that your initial choice of priority is unchanged?

Here, we investigate circumstances in which after such an extinction event, the ranking order obtained from the FP index and thus conservation priorities might radically change. We consider  scenarios where several species go extinct, as well as scenarios where only one species goes extinct. 

The manuscript is organized as follows. We first introduce all relevant concepts and notation. We then show that for each phylogenetic tree, there are edge lengths such that the extinction of one leaf per cherry completely reverses the ranking obtained from the FP index (Theorem \ref{thm_cherries}). Afterwards, we analyze cases in which only one species, e.g. the lowest ranked one, goes extinct, and show that this can already have drastic effects on the prioritization order (Theorem \ref{thm_secondleast_becomes_first} and Theorem \ref{thm_boundoneleaf}).  
After briefly considering the case of ultrametric caterpillar trees, we complement our theoretical results by studying the effects of species extinction for a collection of empirical trees obtained from the TreeBase database (\citet{treebase2,treebase1}) as well as for simulated data. In both cases, we can observe changes in the prioritization orders when species go extinct, indicating that our theoretical results are not merely mathematical artifacts but are of practical importance. However, these studies also show that the rankings become more ``robust'' to species extinctions the larger the tree.
We end by discussing our results and indicating directions for future research.

\section{Definitions and Background}\label{Sec_Preliminaries}

\subsection*{Phylogenetic $X$-trees and related concepts}
Let $X$ denote a non-empty finite set (of taxa) with $|X|=n$. A \emph{rooted binary phylogenetic $X$-tree} $T=(V(T),E(T))$ is a rooted tree (or, more precisely, an \emph{arborescence}) with root vertex $\rho$ of in-degree 0 and out-degree 2, where all edges are directed away from the root, all interior vertices apart from $\rho$ have in-degree 1 and out-degree 2, and the leaves (also referred to as taxa) are bijectively labelled by $X$. For technical reasons, if $|X|=1$, we additionally allow $T$ to consist of a single vertex, which is at the same time the root and only leaf of $T$.
Since all phylogenetic $X$-trees in this paper are rooted and binary, we will often refer to them simply as \emph{phylogenetic trees} or \emph{trees}. Moreover, we call the graph-theoretical tree without leaf labels underlying $T$, the \emph{tree shape} or \emph{topology} of $T$. 
Furthermore, the edges incident to the leaves are referred to as \emph{pendant} edges, whereas all other edges are called \emph{inner} edges.
Additionally, we assume that each edge $e$ of $T$ is assigned a strictly positive edge length $\lambda_T(e) \in \mathbb{R}_+$, representing time or evolutionary distance. Whenever there is no ambiguity we simply refer to the length of an edge $e$ as $\lambda_e$. Moreover, we call $T$ an \emph{ultrametric tree} if the path lengths from the root to all leaves of $T$ are identical. Note that the concept of ultrametric trees is also often referred to as the \emph{molecular clock hypothesis} in biology.

A vertex $v$ of $T$ is a \emph{descendant} of a vertex $u$ of $T$ (and $u$ is an \emph{ancestor} of $v$), if $u$ lies on the unique path from $\rho$ to $v$ in $T$. In particular, a vertex $u$ is the \emph{parent} of a vertex $v$ in $T$, if $v$ is a descendant of $u$ and the edge $e=(u,v)$ exists, i.e. $e=(u,v) \in E(T)$. Then, a \emph{cherry} $[x_i,x_j]$ of $T$ is a pair of taxa $x_i,x_j \in X$ such that $x_i$ and $x_j$ have the same parent in $T$. We use $c_T$ to denote the number of cherries of $T$. A phylogenetic tree that has precisely one cherry is called a \emph{caterpillar tree} (note that up to permuting leaf labels, the caterpillar tree is unique).

When $n \geq 2$, we will also often decompose $T$ into its two \emph{maximal pendant subtrees} $T_a$ and $T_b$ rooted at the children $a$ and $b$ of $\rho$, and we denote this decomposition by $T=(T_a,T_b)$. We use $X_a$ and $X_b$ to denote the leaf sets of $T_a$ and $T_b$, respectively. Moreover, we use $n_a$ and $n_b$ to refer to $|X_a|$ and $|X_b|$, respectively, and assume without loss of generality that $n_a \geq n_b \geq 1$.

Finally, for a subset $Y\subseteq X$, the \textit{induced subtree} ${T_Y}$ of $T$ is the rooted phylogenetic $Y$-tree obtained from the minimal subtree of $T$ connecting the taxa in $Y$ by suppressing all non-root, degree-2 vertices, and adding up the edge lengths of edges that are ``merged'' into a new edge. 
Note that if $T=(T_a,T_b)$ is a tree with $n \geq 2$ leaves and $Y \subseteq X$ is such that $Y \cap X_a = \emptyset$ or $Y \cap X_b = \emptyset$, i.e. $Y$ contains only taxa from $X_a$ or from $X_b$ but not from both, the minimal subtree of $T$ connecting the taxa in $Y$ will still contain the root $\rho$ of $T$ after suppressing all non-root degree-2 vertices, and $\rho$ will have out-degree 1 in this subtree. In particular, there will be a ``root edge'' from $\rho$ to the lowest common ancestor (in $T$) of all taxa in $Y$. In this case, we additionally delete $\rho$ and its incident edge to obtain $T_Y$ in order to ensure that $T_Y$ still is a rooted binary phylogenetic tree according to our definition. However, in the following, we often explicitly enforce induced subtrees $T_Y$ such that $Y$ contains taxa from both $X_a$ and $X_b$ to prevent this from happening.

\subsection*{The Fair Proportion index}
The \emph{Fair Proportion (FP) index} apportions the total sum of edge lengths of $T$ (also referred to as the `phylogenetic diversity' of $X$ (\citet{Faith1992})) among the taxa in $X$ (\citet{Redding2003, Isaac2007}).
More precisely, the FP index for $x \in X$ 
is defined as 
\begin{equation}\label{Def_FP}
    FP_T(x) = \sum\limits_{e \in P(T;\rho,x)} \frac{\lambda_e}{D_e},
\end{equation}
where $P(T; \rho,x)$ denotes the path in $T$ from the root to leaf $x$ and $D_e$ is the number of leaves descended from the edge $e$. Essentially, the FP index distributes each edge length equally among descending leaves. It is thus not hard to see that $\sum\limits_{x \in X} FP_T(x) = \sum\limits_{e \in E(T)} \lambda_e$.  As an example, for tree $T$ depicted in Figure \ref{fig_goodexample} and taxon $x_1$, we have $FP_T(x_1)=\frac{1}{6} + 61 = 61 \frac{1}{6}$.

\subsection*{Strict and reversible rankings} 
Recall that a \emph{ranking} $\pi(S,f)$ for a set $S=\{s_1,\ldots,s_n\}$ based on a function $f:S\rightarrow \mathbb{R}$ is an ordered list of the
elements of $S$ such that $f(s_i)\geq f(s_j)$ if and only if $s_i$
appears before $s_j$ in $\pi$. A ranking function $f$ is called \emph{strict} if it is one-to-one, that is, if there are no
ties.  In other words, $f$ is strict if $f(s_i)\neq f(s_j)$ for all $i\neq j$.
Note that in the following we mostly consider rankings $\pi(X,FP_T)$, where $T$ is a phylogenetic $X$-tree. Therefore, whenever there is no ambiguity, we use the shorthand $\pi_T$ instead of $\pi(X,FP_T)$.

We are interested in diversity rankings whose order is reversed if a species or set of species goes extinct.  Formalizing this intuition, we call a ranking $\pi_T$ \emph{reversible} if there is a subset $X'$ of $X$ whose removal from $X$ and from $T$ leads to an induced subtree $\widetilde{T}=T_{\widetilde{X}}$ of $T$ (with $\widetilde{X}:=X\setminus X'$) whose corresponding ranking $\pi_{\widetilde{T}}=\pi(\widetilde{X},FP_{\widetilde{T}})$ ranks the species in the opposite order to $\pi_T$. In particular, if $FP_{\widetilde{T}}(x_i)>FP_{\widetilde{T}}(x_j)$ in $\pi_{\widetilde{T}}$, then $FP_T(x_i)<FP_T(x_j)$ in $\pi_T$. 

An example of a tree $T$ with edge lengths that induce a strict and reversible ranking is given in Figure \ref{fig_goodexample}.

\begin{figure}[htbp]
\flushleft
\begin{minipage}{0.65\textwidth}
\includegraphics[scale=.225]{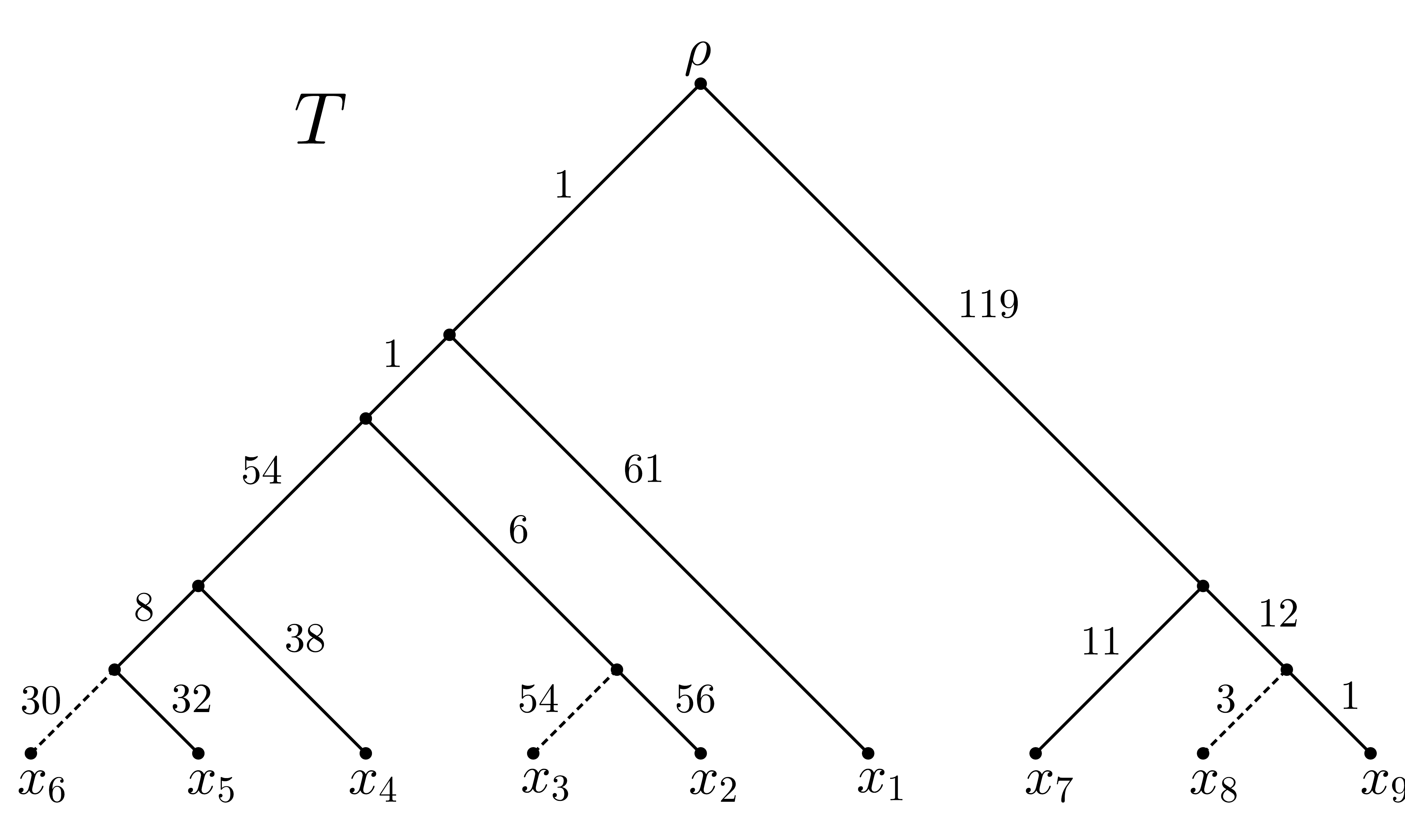}
\hfill
\end{minipage}
\begin{minipage}{0.25\textwidth}
\begin{tabular}{ccc}
\toprule
$x$ & $FP_T(x)$ & $FP_{\widetilde T}(x)$ \\
\midrule
$x_1$ & 61.17 & 61.25 \\
$x_2$ & 59.37 & 62.58 \\
$x_3$ & 57.37 & --\\
$x_4$ & 56.37 & 65.58\\
$x_5$ & 54.37 & 67.58 \\
$x_6$ & 52.37 & -- \\
$x_7$ & 50.67 & 70.50 \\
$x_8$ & 48.67 & --\\
$x_9$ & 46.67 & 72.50\\
\bottomrule
\end{tabular}
\end{minipage}
\caption{A phylogenetic tree $T$ on $X=\{x_1, \ldots, x_9\}$, with edge lengths that induce a strict and reversible ranking when leaves $x_3$, $x_6$ and $x_8$ (i.e. one leaf per cherry) are deleted from $T$ to form $\widetilde{T}$.
More precisely, $FP_{T}(x_1) >FP_{T}(x_2) > \ldots > FP_{T}(x_9) $, but when leaves $x_3$, $x_6$ and $x_8$ are deleted, we have  $FP_{\widetilde{T}}(x_1) <  FP_{\widetilde{T}}(x_2) < \ldots <FP_{\widetilde{T}}(x_9)$. Thus, the induced ranking $\pi_{T}$ is strict and reversible.}
\label{fig_goodexample}
\end{figure}

\subsection*{Kendall's $\tau$ rank correlation coefficient}
The \emph{Kendall's $\tau$ rank correlation coefficient} can be used to quantify the similarity and association of ranked data obtained from different ranking functions. Let $\pi_1(S,f_1)$ and $\pi_2(S,f_2)$ be two rankings for a set $S = \{s_1, \ldots, s_n\}$. Note that these rankings do not need to be strict, but may contain ties. We say that a pair $(s_i,s_j)$ of elements from $S$ with $i < j$ is a \emph{concordant} pair if $s_i$ and $s_j$ are in the same order in $\pi_1$ and $\pi_2$ (i.e., $f_1(s_i) > f_1(s_j)$ and $f_2(s_i) > f_2(s_j)$; or $f_1(s_i) < f_1(s_j)$ and $f_2(s_i) < f_2(s_j)$). On the other hand, if $s_i$ and $s_j$ are in the opposite order in $\pi_1$ and $\pi_2$ (i.e., $f_1(s_i) > f_1(s_j)$ and $f_2(s_i) < f_2(s_j)$; or $f_1(s_i) < f_1(s_j)$ and $f_2(s_i) > f_2(s_j)$), then $(s_i,s_j)$ is called a \emph{discordant} pair. Let $n_c$ denote the number of concordant pairs, let $n_d$ denote the number of discordant pairs, and let $n_{\pi_1}$ and $n_{\pi_2}$ denote the number of pairs that are tied only in $\pi_1$ or only in $\pi_2$, respectively (if a tie occurs for the same pair in both $\pi_1$ and $\pi_2$, it is not added to neither $n_{\pi_1}$ nor $n_{\pi_2}$).
Then, the Kendall's $\tau$ coefficient (or Kendall's $\tau_b$ as this version of the coefficient, which allows for ties, is often called) is defined as
\begin{align*}
    \tau =\frac{n_c-n_d}{\sqrt{(n_c+n_d+n_{\pi_1})(n_c+n_d+n_{\pi_2})}}.
\end{align*}
Note that $\tau \in [-1,1]$, where the rankings are the same if $\tau=1$ and they are completely reversed if $\tau=-1$. If $\tau=0$, the rankings are uncorrelated.

As an example, for tree $T$ on $X=\{x_1, \ldots, x_9\}$ depicted in Figure \ref{fig_goodexample} and tree $\widetilde{T}$ on $\widetilde{X}=X \setminus \{x_3,x_6,x_8\}$ obtained from $T$ by deleting leaves $x_3$, $x_6$, and $x_8$, we have $\pi_T(\widetilde{X}, FP_T) = (x_1, x_2, x_4, x_5, x_7, x_9)$ and $\pi_{\widetilde{T}}(\widetilde{X}, FP_{\widetilde{T}}) = (x_9, x_7, x_5, x_4, x_2, x_1)$. As $|\widetilde{X}|=6$, there are $\binom{6}{2}=15$ possible pairs and all of them are discordant. Moreover, both rankings are strict. Thus, $n_d=15$ and $n_c=n_{\pi_T} = n_{\pi_{\widetilde{T}}}=0$, and we have 
$$ \tau = \frac{0-15}{\sqrt{(0+15+0)(0+15+0)}} = -1.$$

\section{Results}
We are now in the position to study the effects of species extinction, i.e. leaf deletions, on rankings obtained from the FP index. We begin by considering circumstances that lead to strict and reversible rankings, i.e. circumstances in which the extinction of species completely reverses conservation priorities. 

\subsection{Extinction scenarios completely reversing conservation priorities} \label{Sec_CherryExtinction}

We start by showing that a strict ranking can only be reversed if the set $X'$ of deleted leaves contains at least one leaf of each cherry.

\begin{thm}\label{thm_cherries1} Let $T$ be a phylogenetic $X$-tree with $|X|\geq 3$. Let $ \pi_T$ be a strict and reversible ranking for $T$ with respect to $X'\subset X$ and induced subtree $\widetilde{T}$ on taxon set $\widetilde{X}=X\setminus X'$. Let $c=[x_i,x_j]$ be a cherry of $T$. Then, $X'$ contains at least one of the elements $x_i$, $x_j$, i.e. $|X'\cap \{x_i,x_j\}|\geq 1$. 
\end{thm}

The proof of this theorem is provided in the appendix. As an illustration, consider Figure \ref{fig_goodexample}. Here, $T$ induces a strict and reversible ranking when one leaf per cherry is deleted (in this case, leaves $x_3$, $x_6$, and $x_8$ are deleted). If we had kept at least one of those leaves, say $x_8$, the resulting ranking would not have been reversible. More explicitly, if we let $\widehat{T}$ denote the phylogenetic $\widehat{X}$-tree with $\widehat{X}=X \setminus \{x_3,x_6\}$, we have $\pi_T(\widehat{X}, FP_T)=(x_1,x_2,x_4,x_5,x_7,x_8,x_9)$ and $\pi_{\widehat{T}}(\widehat{X},FP_{\widehat{T}}) = (x_5,x_4,x_2,x_1,x_7,x_8,x_9)$, which shows that the induced rankings are not completely reversed.\\

So if we want to find a strict and reversible ranking $\pi_T$, then at least one leaf per cherry of $T$ has to be deleted; otherwise no suitable edge lengths for $T$ can exist that induce such a ranking. The following fundamental theorem, however, shows that this necessary condition is even sufficient: For each phylogenetic tree, deleting one leaf per cherry is sufficient for the existence of edge lengths that induce a strict and reversible ranking. Moreover, we can even ensure that the species that has the highest FP index in $T$ is still present in $\widetilde{T}$, where it will have the lowest FP index (as $\pi_T$ is strict and reversible). 

We formalize this in the following main theorem of this section. 

\begin{thm}\label{thm_cherries}
Let $n\geq 2$ and let $T=(T_a,T_b)$ be a rooted binary phylogenetic $X$-tree with $|X|=n$. 
Let $\widetilde{T}$ be the induced subtree on leaf set $\widetilde{X}\subset X$ that results from $T$ when we delete one leaf out of each cherry of $T$ and suppress the resulting vertices of in-degree 1 and out-degree 1. 

Then there exist strictly positive edge lengths $\lambda_1,\ldots,\lambda_{2n-2}$ for $T$ such that there is a strict ranking $\pi_T$ for the leaves of $T$ concerning the FP index which is reversible with respect to $\widetilde{T}$ and such that $\widetilde{T}$ contains the species which has the highest FP index in $T$ and such that the species with the lowest FP index in $T$ is \emph{not} contained in $\widetilde{T}$. 

In particular, if $c_T$ denotes the number of cherries of $T$ and if $x_1 \in X$ is such that $x_1=\argmax\limits_{x \in X}FP_T(x)$, then we have $FP_T(x_1)>FP_T(x_2)>\ldots > FP_T(x_n)$ and $FP_{\widetilde{T}}(x_1)<FP_{\widetilde{T}}(\widetilde{x}_2)< \ldots < FP_{\widetilde{T}}(\widetilde{x}_{n-c_T})$, where $x_1$ as well as $\widetilde{x}_i$ are contained in $ \widetilde{X}$ for all $i=2,\ldots,n-c_T$ and where $FP_T(\widetilde{x}_i)>FP_T(\widetilde{x}_j)$ if and only if $FP_{\widetilde{T}}(\widetilde{x}_i)<FP_{\widetilde{T}}(\widetilde{x}_j)$. Moreover, if $x' \in X$ is such that $x'=\argmin\limits_{x \in X}FP_T(x)$, then $x' \not\in \widetilde{X}$.
\end{thm}

The proof of this theorem (together with additional lemmas required for the proof) is provided in the appendix and uses induction on the number of leaves. However, we remark that it is constructive in the following sense: If $T=(T_a,T_b)$ is a phylogenetic tree with $n \geq 2$ leaves such that both $T_a$ and $T_b$ induce strict and reversible rankings, then the proof of Theorem \ref{thm_cherries} establishes a technique to construct a strict and reversible ranking for $T$ by suitably modifying the edge lengths of $T_a$ and $T_b$. By recursively applying this technique, an edge length assignment yielding a strict and reversible ranking can be found for any given phylogenetic tree $T$, regardless of the number of leaves or shape of $T$. \\

Note that while, by Theorem \ref{thm_cherries1}, several leaves of $T$ need to be deleted in order to reverse the entire ranking if $T$ contains more than one cherry, the following corollary shows that for all values of $n$, the extinction of only \emph{one} species, even the one with the lowest FP index, may be sufficient to cause a strict and reversible ranking (depending on the tree shape).

\begin{cor} \label{cor_cat} Let $n \in \mathbb{N}_{\geq 2}$. Then, there exists a rooted binary phylogenetic $X$-tree $T$ with $|X|=n$, namely the caterpillar tree, and edge lengths for $T$, such that $T$ has a strict and reversible ranking with respect to $\widetilde{X}:=X\setminus \{x\}$, where $x=\argmin\limits_{x' \in X}FP_T(x')$, i.e. $x \in X$ is the species with the lowest $FP_T$ value. 
\end{cor}

\begin{proof} Let $T$ be the caterpillar tree on $n$ leaves, i.e. $T$ has precisely one cherry. By Theorem \ref{thm_cherries} there are edge lengths for $T$ which assign the smallest $FP_T$ value to a leaf in a cherry such that, if we delete this leaf, the entire ranking induced by $FP_T$ gets reversed. 
\end{proof}

\subsection{The impact of the extinction of a single species}\label{Sec_SingleExtinction} 
While Corollary \ref{cor_cat} implies that for the caterpillar tree the extinction of a single species may completely reverse conservation priorities, for trees that contain more than one cherry the extinction of a single species cannot completely reverse the ordering (due to Theorem \ref{thm_cherries1}). In the following we show, however, that the extinction of a single species, even if it is ranked lowest, can still cause radical changes in conservation priorities. We begin by showing that given a phylogenetic tree $T$, the extinction of the species with the lowest $FP_T$ value can have the effect that the species with the second lowest $FP_T$ value has the highest $FP_{\widetilde{T}}$ value. 

\begin{thm}\label{thm_secondleast_becomes_first}
Let $T$ be a rooted binary phylogenetic $X$-tree with $|X|=n \geq 2$. 
Then, there exist strictly positive edge lengths $\lambda_1, \ldots, \lambda_{2n-2}$ for $T$ such that 
\begin{enumerate}[(i)]
\item the ranking $\pi_T$ induced by the FP index for the leaves of $T$ is strict, and 
\item deleting leaf $y \coloneqq \argmin\limits_{x \in X} FP_T(x)$ from $T$ results in a strict ranking $\pi_{\widetilde{T}}$ for tree $\widetilde{T} \coloneqq T \setminus \{y\}$ on leaf set $\widetilde{X} = X \setminus \{y\}$, for which $w \coloneqq \argmax\limits_{x \in \widetilde{X}} FP_{\widetilde{T}}(x) = \argmin\limits_{x \in \widetilde{X}} FP_T(x)$.
\end{enumerate}
In other words, there exist strictly positive edge lengths for $T$ such that if the species with the lowest $FP_T$ value goes extinct, the species with the second lowest $FP_T$ value has the highest $FP_{\widetilde{T}}$ value. 
\end{thm}

The proof of this theorem is provided in the appendix, but an example for its implications is depicted in Figure \ref{fig:lowestdeletion}. Here, the species with the lowest $FP_T$ value is taxon $x_1$, and the one with the second lowest $FP_T$ value is taxon $x_2$. However, when $x_1$ goes extinct, $x_2$ is the taxon with the highest FP index in the remaining tree. Note that the overall ranking is not completely reversed in this situation, in accordance with Theorem~\ref{thm_cherries1}.

\begin{figure}[htbp]
\flushleft
\begin{minipage}{0.6\textwidth}
\includegraphics[scale=0.3]{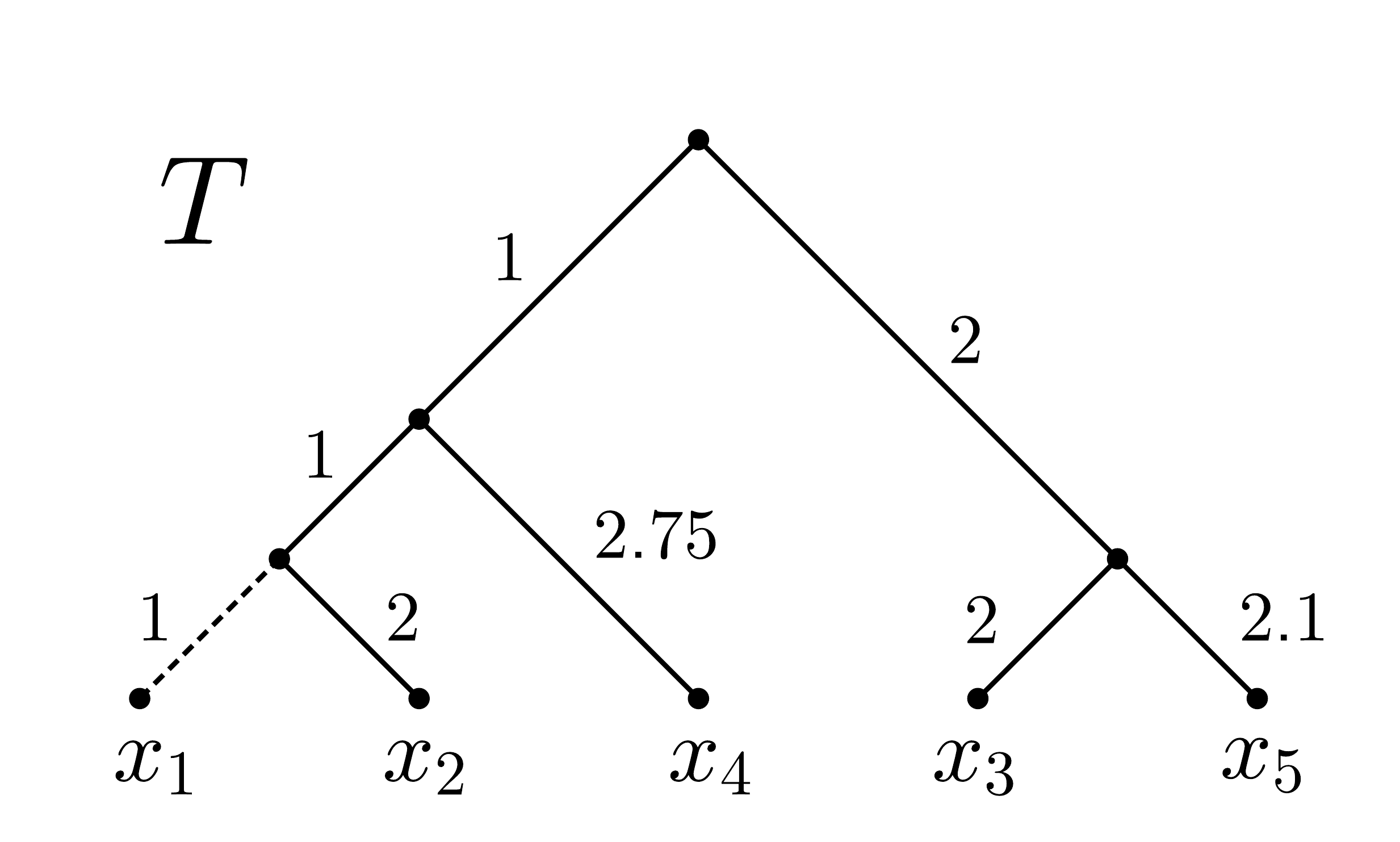}
\end{minipage}
\begin{minipage}{0.3\textwidth}
\begin{tabular}{ccc}
\toprule
$x$ & $FP_T(x)$ & $FP_{\widetilde T}(x)$ \\
\midrule
$x_1$ & $1.83$ & -- \\
$x_2$ & $2.83$ & $3.5$\\
$x_3$ & $3$ & $3$ \\
$x_4$ & $3.08$ & $3.25$\\
$x_5$ & $3.1$ & $3.1$ \\
\bottomrule
\end{tabular}
\end{minipage}
\caption{Tree $T$ is an example of a tree as described in Theorem \ref{thm_secondleast_becomes_first}. Here, $x_1$ has the lowest $FP_T$ value and $x_2$ the second lowest, but $x_2$ has the highest $FP_{\widetilde{T}}$ value in the remaining subtree $\widetilde{T}$ when $x_1$ is deleted.}
\label{fig:lowestdeletion}
\end{figure}

As the extinction of the lowest ranked species can have the effect that the formerly second least important species is ranked highest when the FP indices are re-computed, conservation efforts might need to be reallocated to focus on this species.
In the following, we analyze the extinction of any single species, not necessarily (but possibly) the lowest ranked one, in a little more depth. We first show that if only one species goes extinct and this species is distinct from the highest ranked species, say $x^\ast$, then while there can be changes in conservation priorities, we can at least bound the number of species that will receive a higher FP index and thus require more urgent conservation attention than $x^\ast$ in the remaining tree (Theorem \ref{thm_boundoneleaf}(1)). More precisely, this number is bounded by $n_a-1$, where $n_a$ is the number of leaves in the larger subtree $T_a$ of a rooted binary tree $T=(T_a,T_b)$ with $n \geq 3$ leaves. 
Note that $\lfloor \frac{n-1}{2} \rfloor \leq n_a-1 \leq n-2$. The lower bound is achieved if $T=(T_a,T_b)$ is such that $n_a$ and $n_b$ differ by at most one (i.e., $n_a = \lceil \frac{n}{2} \rceil$ and $n_b = \lfloor \frac{n}{2} \rfloor$), i.e., the number of leaves of $T$ is as evenly distributed across $T_a$ and $T_b$ as possible. The upper bound, on the other hand, is achieved if the difference in the number of leaves between $T_a$ and $T_b$ is as large as possible (i.e., $n_a=n-1$ and $n_b=1$) as for example in the case of the caterpillar tree. This means that the impact of a single species extinction on the conservation priorities of the remaining species directly depends on the shape of the underlying tree and how different its subtree sizes are. However, we then also show that the bound of $n_a-1$ species receiving a higher FP index than the formerly highest ranked species $x^\ast$ can be realized in all cases (Theorem \ref{thm_boundoneleaf}(2)). Thus, in particular if $n_a$ is large, the effect of a single species extinction might require a drastic shift in conservation attention for the remaining species. However, if $n_a$ is small, then the impact of a single species extinction might be considered less dramatic (even though it could still be the case that almost half of the species require more urgent conservation attention than $x^\ast$).

\begin{thm} \label{thm_boundoneleaf} Let $T=(T_a,T_b)$ be a rooted binary phylogenetic $X$-tree with $|X|=n\geq 3$ such that $T_a$ and $T_b$ have $n_a$ and $n_b$ leaves, respectively, where $n_a\geq n_b$. Let $f:E(T)\longrightarrow \mathbb{R}_+$ be some function that assigns all edges of $T$ positive edge lengths. Let $x^*:=\argmax\limits_{x \in X} FP_T(x) $ be the leaf of $T$ with the highest $FP_T$ value concerning the edge length assignment of $f$. Then, we have: \begin{enumerate}
    \item If a leaf $x'$ other than $x^*$ is deleted (e.g. the leaf with minimal $FP_T$ value) to derive a tree $\widetilde{T}$, we denote the number of leaves that have a higher $FP_{\widetilde{T}}$ value than $x^*$ in $\widetilde{T}$ by $\mathcal{N}$ and have that $\mathcal{N}\leq n_a-1$ (with $ \lfloor \frac{n-1}{2}\rfloor \leq n_a-1$).
\item There exists an edge length assignment $\widehat{f}$ such that this bound is achieved, i.e. $\mathcal{N}= n_a-1$, and the resulting ranking $\pi_T$ is strict.
\end{enumerate}
\end{thm}

The proof of this theorem can again be found in the appendix. For the first part we use the fact that if $T=(T_a,T_b)$ is a tree with $n \geq 3$ leaves and some (but not all) leaves from only one of $T_a$ and $T_b$, say $T_a$, are deleted, the FP index for all taxa in $T_b$ remains the same (Lemma \ref{lem_nonaffectedsubtree} in the appendix), whereas the FP index for all taxa in $T_a$ strictly increases (Lemma \ref{lem_affectedsubtree} in the appendix). The second part of the proof is similar to the proof of Theorem \ref{thm_cherries}. In particular, it provides a constructive way to find an edge length assignment with the claimed properties. 

An example to illustrate Theorem \ref{thm_boundoneleaf} is given in Figure \ref{fig:lowest_higher}. Here, the taxon with the highest FP index is taxon $x_6$. If we now delete taxon $x_1$, which has the lowest FP index in $T$, we have that $n_a-1=3$ leaves receive a higher FP index than $x_6$ in the resulting tree $\widetilde{T}$.

\begin{figure}[htbp]
\flushleft
\begin{minipage}{0.6\textwidth}
\includegraphics[scale=0.3]{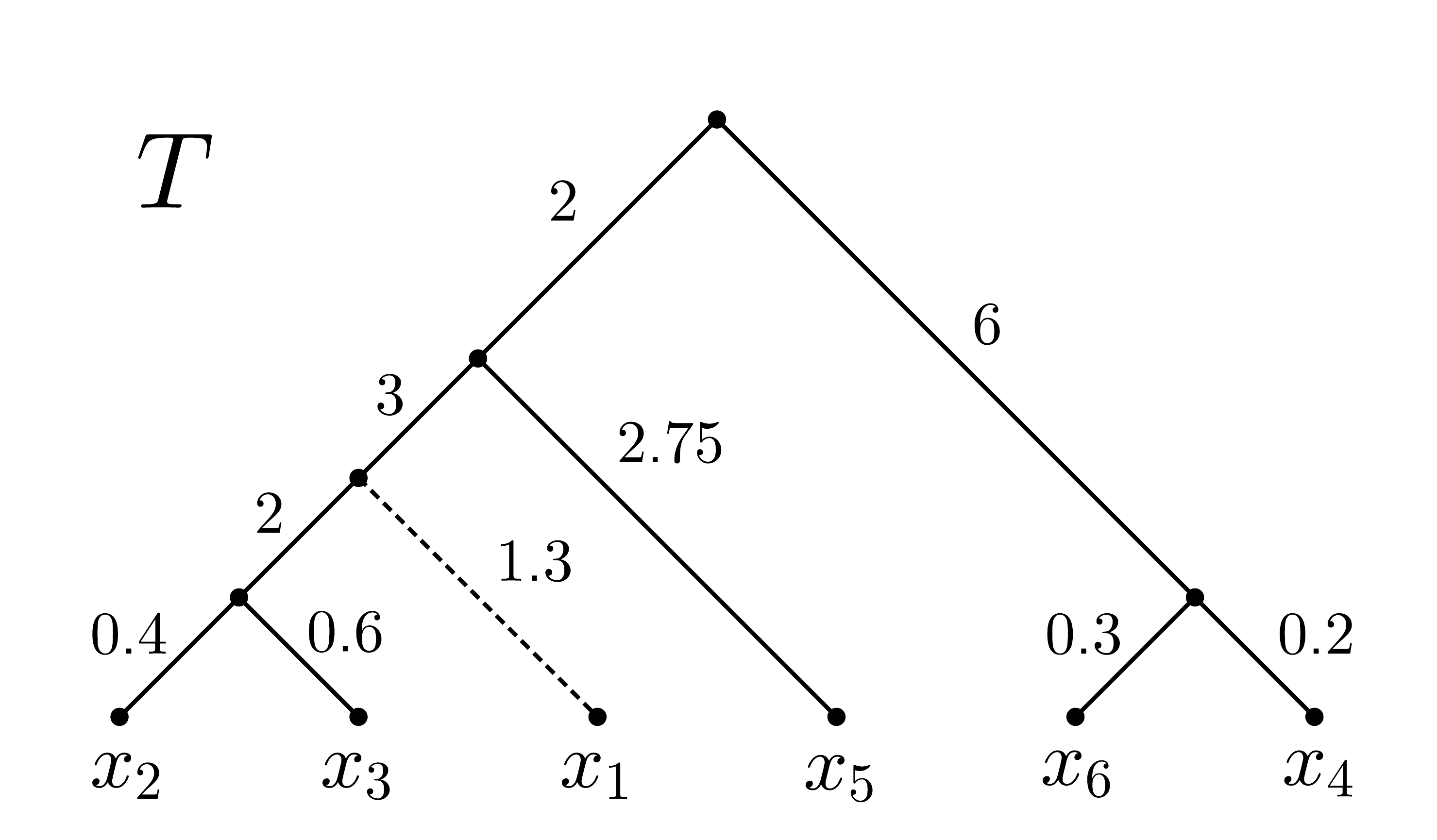}
\end{minipage}
\begin{minipage}{0.15\textwidth}
\begin{tabular}{ccc}
\toprule
$x$ & $FP_T(x)$ & $FP_{\widetilde T}(x)$ \\
\midrule
$x_1$ & $2.8$ & -- \\
$x_2$ & $2.9$ & $3.57$ \\
$x_3$ & $3.1$ & $3.77$  \\
$x_4$ & $3.2$ & $3.2$ \\
$x_5$ & $3.25$ & $3.42$\\
$x_6$ & $3.3$ & $3.3$ \\
\bottomrule
\end{tabular}
\end{minipage}
\caption{Tree $T$ is an example for a tree as described in Theorem \ref{thm_boundoneleaf}. Here, $x^\ast = x_6$ has the highest FP index, but if $x'=x_1$ (the leaf with the lowest FP index) is deleted, $n_a-1=3$ leaves (namely, $x_2$, $x_3$, and $x_5$) have a higher FP index than $x_6$ in the resulting tree $\widetilde{T}$.}
\label{fig:lowest_higher}
\end{figure}

\subsection{The impact of the edge lengths} \label{Subsec_Ultrametric}
Note that the proofs of the preceding theorems rely on a careful choice of edge lengths. It is thus a natural question to analyze how restrictions on the edge lengths influence the results. For instance, if we assume a molecular clock condition, i.e. if we restrict the analysis to ultrametric trees where all leaves have the same distance to the root, what are the worst-case scenarios in this setting? In the special case of caterpillar trees, a molecular clock assumption is beneficial in the sense that the extinction of one or more leaves does not change the ranking order of the remaining leaves. 

\begin{prop} \label{Prop_caterpillar_ultrametric}
Let $T$ be a rooted binary ultrametric caterpillar tree on $X$ with $|X|=n \geq 3$, and let $X' \subset X$ be a subset of the leaves. Let $\widetilde{T}$ be the induced subtree of $T$ restricted to the leaves in $\widetilde{X} \coloneqq X \setminus X'$. Then, $FP_T(x_i) \geq FP_T(x_j)$ implies $FP_{\widetilde{T}}(x_i) \geq FP_{\widetilde{T}}(x_j)$ for all $x_i,x_j \in \widetilde{X}$.
\end{prop}

We provide a proof of this proposition in the appendix. Intuitively, in an ultrametric caterpillar tree, the fewer edges separate a leaf from the root (i.e., the smaller the so-called \emph{depth} of a leaf), the higher its FP index (in particular, the two leaves in the cherry of a caterpillar tree have the lowest FP index, and the leaf that is adjacent to the root has the highest FP index). Now, if one or more leaves are deleted from a caterpillar tree, the resulting tree is again a caterpillar tree, for which this property still holds. An example is given in Figure \ref{fig:tcat}.

\begin{figure}[htbp]
\flushleft
\begin{minipage}{0.6\textwidth}
\includegraphics[scale=0.3]{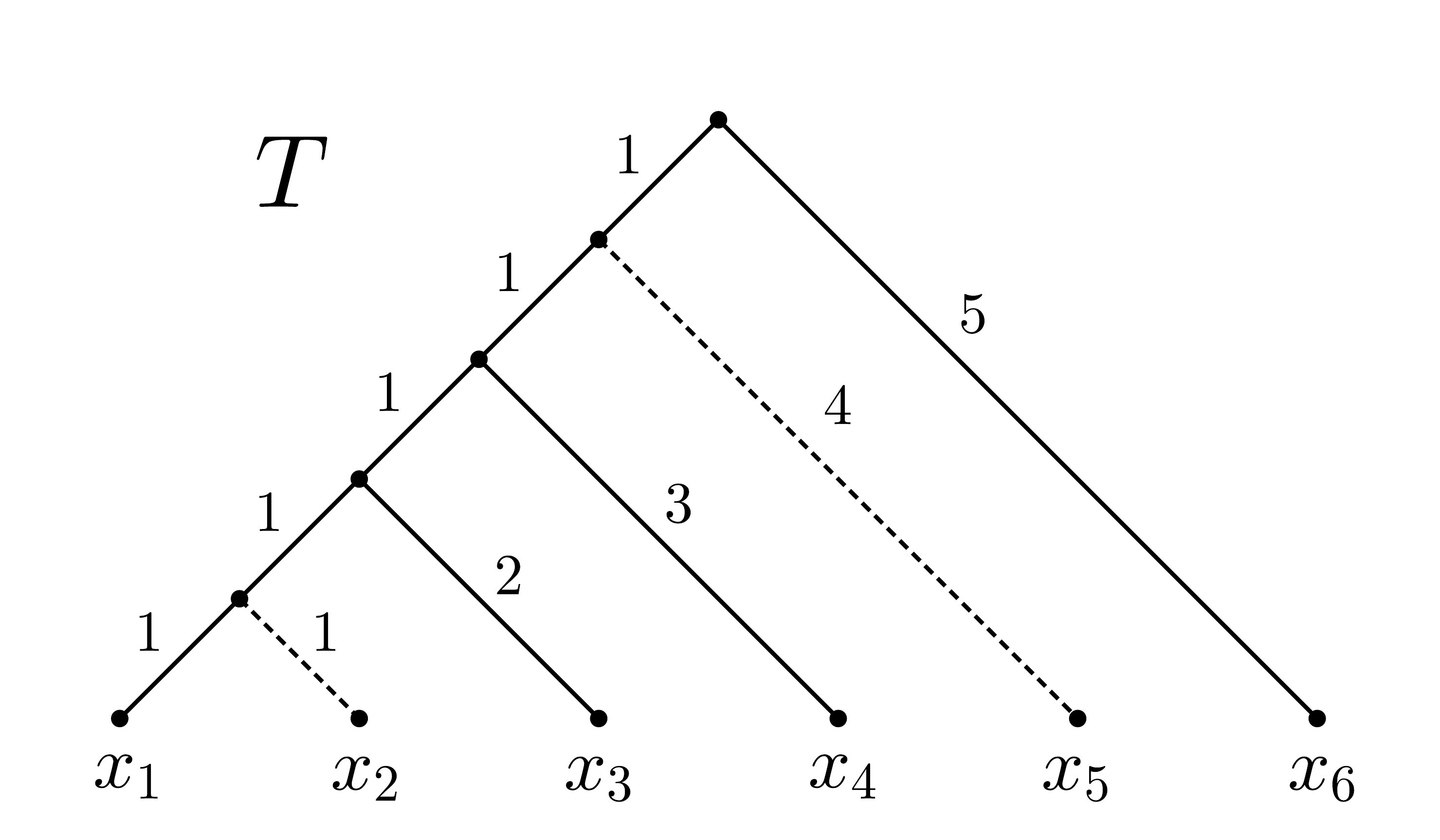}
\end{minipage}
\begin{minipage}{0.15\textwidth}
\begin{tabular}{ccc}
\toprule
$x$ & $FP_T(x)$ & $FP_{\widetilde T}(x)$ \\
\midrule
$x_1$ & $2.28$ & $3.17$\\
$x_2$ & $2.28$ & -- \\
$x_3$ & $2.78$ & $3.17$ \\
$x_4$ & $3.45$ & $3.67$ \\
$x_5$ & $4.2$ & -- \\
$x_6$ & $5$ & $5$ \\
\bottomrule
\end{tabular}
\end{minipage}
\caption{Ultrametric caterpillar tree $T$ with $\pi_T=(x_6,x_5,x_4,x_3,x_2,x_1)$. If leaves $x_2$ and $x_4$ are deleted, the resulting tree $\widetilde{T}$ is again an ultrametric caterpillar tree and we have $\pi_{\widetilde{T}} = (x_6,x_4,x_3,x_1)$. In particular, the remaining leaves appear in the same order in $\pi_{\widetilde{T}}$ as in $\pi_T$, i.e. the ranking is not changed by leaf deletions.}
\label{fig:tcat}
\end{figure}

So, in case of ultrametric caterpillar trees, the extinction of species does not influence the ranking order of the remaining leaves. However, as Figure \ref{Fig_UltrametricTree} shows, the assumption of a molecular clock does not always imply that the ranking order is unaffected by leaf deletions. This also becomes evident in our simulation study below. 
A more in-depth analysis of the effects of leaf deletions on the rankings induced by the FP index in the case of ultrametric trees is thus an interesting direction for future research.

\begin{figure}[htbp]
\centering
\begin{minipage}{0.5\textwidth}
\includegraphics[scale=0.3]{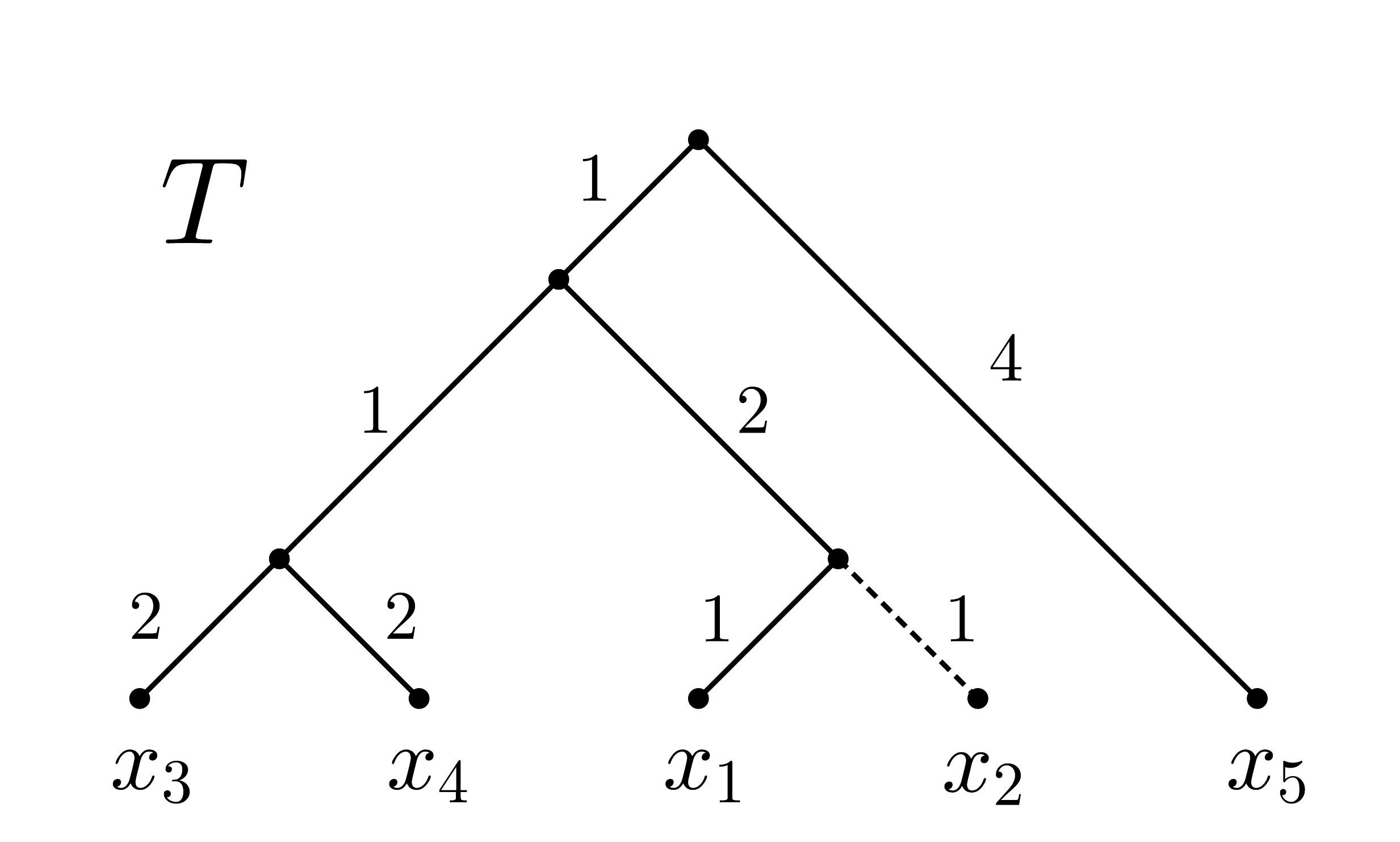}
\end{minipage}
\begin{minipage}{0.25\textwidth}
\begin{tabular}{ccc}
\toprule
$x$ & $FP_T(x)$ & $FP_{\widetilde T}(x)$ \\
\midrule
$x_1$ & $2.25$ & $3.33$\\
$x_2$ & $2.25$ &  -- \\
$x_3$ & $2.75$ & $2.83$ \\
$x_4$ & $2.75$ & $2.83$ \\
$x_5$ &  $4$ & $4$ \\
\bottomrule
\end{tabular}
\end{minipage}
\caption{Ultrametric tree $T$ with $FP_T(x_1) = FP_T(x_2)=2.25 < FP_T(x_3) = FP_T(x_4) = 2.75 < FP_T(x_5)=4$. If we now construct a tree $\widetilde{T}$ by deleting leaf $x_2$, we get $FP_{\widetilde{T}}(x_3) = FP_{\widetilde{T}}(x_4) = 2.83 < FP_{\widetilde{T}}(x_1) = 3.33 < FP_{\widetilde{T}}(x_5) = 4$. In particular, $x_3$ and $x_4$ are ranked higher than $x_1$ in $T$, while they are ranked lower than $x_1$ in $\widetilde{T}$.}
\label{Fig_UltrametricTree}
\end{figure}

\subsection{Data analysis}\label{subsec:data}

It is conceivable that the changes in FP rankings due to extinctions that we have studied here are just ``theoretical'', and that the problem is not significant with empirical tree data.  In order to test the extent of these problems with empirical data, we accessed the free TreeBase database (\citet{treebase1,treebase2}) on May 14th, 2021 and downloaded all 19,488 trees with up to 100 taxa. We then filtered these trees as follows: We omitted all trees which are unrooted, non-binary or which do not have branch lengths provided for each edge. We also omitted trees for which all branch lengths are 0. The remaining tree set contained 575 trees. For all these trees, we performed the following two analyses with the computer algebra system Mathematica (\citet{mathematica}).

\begin{enumerate}
\item We calculated the FP values for all taxa, subsequently detected the taxon with the lowest FP value, and deleted it from the list of FP values. The resulting list of FP values was saved as $originalList_1$. Then, we deleted the taxon with the lowest FP value and its pending branch also from the original tree. We then re-calculated the FP values for the resulting tree. The FP values of this tree were saved as $newList_1$. Then we calculated Kendall's $\tau$: $\tau(originalList_1,newList_1)$ and saved it in the list $kendall_1$. The results of this first study are shown in Figure \ref{fig:treebasestudy} (black boxplots).

\begin{figure}[htbp]
\centering
\includegraphics[scale=1]{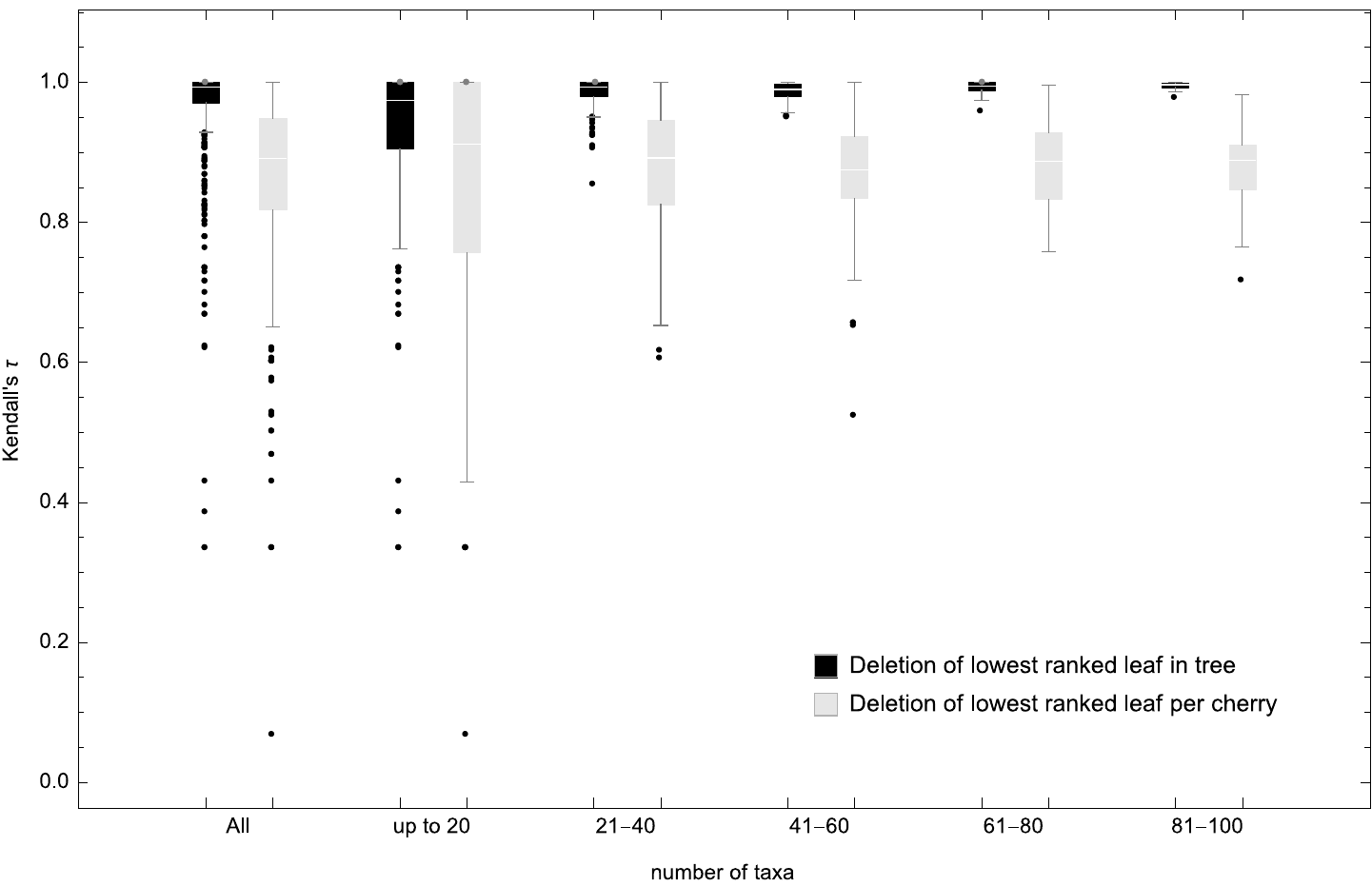}
\caption{The resulting values of Kendall's $\tau$ when the leaf with the smallest FP value gets deleted and all FP values get re-calculated (black  boxplots) and the resulting values of Kendall's $\tau$ when of each cherry the leaf with the smallest FP value gets deleted and all FP values get re-calculated (gray boxplots). The first boxplots contain all trees with up to 100 taxa from TreeBase, whereas the other boxplots are sorted by the numbers of taxa the respective trees contain.}
\label{fig:treebasestudy}
\end{figure}

\item We calculated the FP values for all taxa and subsequently detected all cherries. Then, for each cherry we deleted the taxon with the lowest FP value within the cherry from the list of FP values. The resulting list of FP values was saved as $originalList_2$.
Afterwards, we also deleted these taxa (the lowest ranked one of each cherry) and their pending branch from the original tree and re-calculated the FP values for the resulting tree. The FP values of this tree were saved as $newList_2$. Then we calculated Kendall's tau: $\tau(originalList_2,newList_2)$ and saved it in the list $kendall_2$. The results of this study are shown in Figure \ref{fig:treebasestudy} (gray boxplots).
\end{enumerate}

In both studies, it can be seen that the effects of taxon deletion tend to be less extreme if the tree has more taxa. However, this was to be expected, because, say, a single rank swap between two entries would have a larger effect on Kendall's $\tau$ of a tree with few taxa than on Kendall's $\tau$ of a tree with many taxa.

Recall that Kendall's $\tau$ is 1 precisely if the compared rankings are identical, whereas  Kendall's $\tau$ is 0 if the compared rankings are uncorrelated. Interestingly, for small trees with up to 20 taxa, some outliers are actually closer to 0 than to 1, i.e. their rankings change significantly. One such example is tree \textsf{Tr66501} from TreeBase, which is depicted in Figure~\ref{fig:Tr66501}. In this tree, if you delete taxon \textsf{Arabidopsis\_thaliana\_ANAC019\_At1g52890\_1}, which is the one with the lowest FP value and, as the tree only has one cherry, also the only taxon that gets deleted in both  studies, taxon \textsf{Arabidopsis\_thaliana\_ANAC055\_At3g15500\_1}, which was first ranked lower than taxon \textsf{Arabidopsis\_thaliana\_ANAC072\_At4g27410\_1}, is now ranked the highest. Kendall's $\tau$ in this case gives a value of 0.333, which is at the same time the minimum value observed in the first study.

\begin{figure}[htbp]
\centering
\includegraphics[scale=0.25]{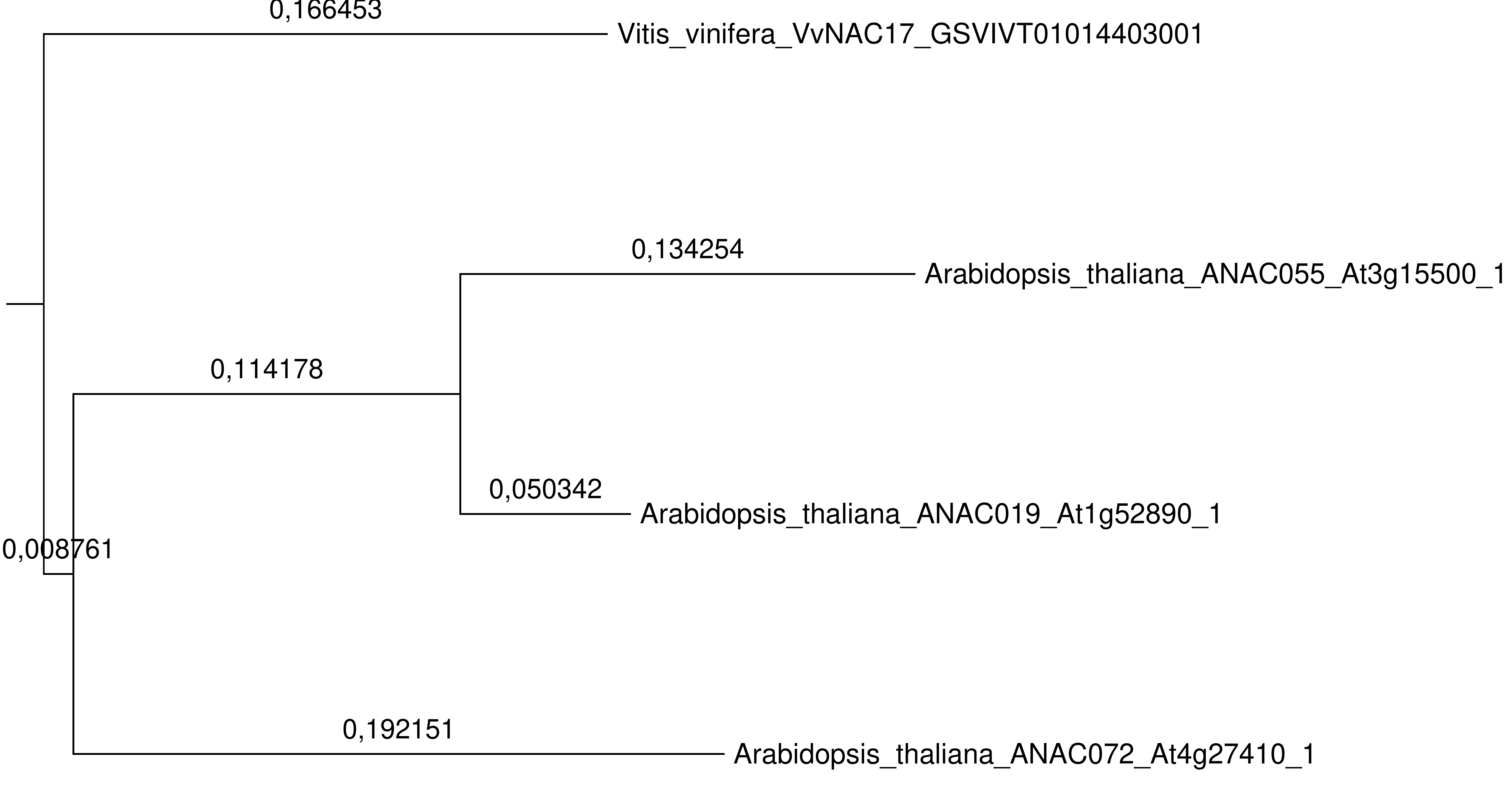}\\[2em]
\begin{tabular}{lcc}
\toprule
$x$ & $FP_{\mathsf{Tr66501}}(x)$ & $FP_{\widetilde{\mathsf{Tr}}\mathsf{66501}}(x)$ \\
\midrule
\textsf{\small Arabidopsis\_thaliana\_ANAC019\_At1g52890\_1} & 0.110351 & -- \\
\textsf{\small Vitis\_vinifera\_VvNAC17\_GSVIVT01014403001} & 0.116453 & 0.116453 \\
\textsf{\small Arabidopsis\_thaliana\_ANAC055\_At3g15500\_1} & 0.194263  &  0.252813 \\
\textsf{\small Arabidopsis\_thaliana\_ANAC072\_At4g27410\_1 } & 0.195071 & 0.196532 \\
\bottomrule
\end{tabular}
\caption{Tree \textsf{Tr66501} from TreeBase. When taxon \textsf{Arabidopsis\_thaliana\_ANAC019\_At1g52890\_1} gets deleted, this changes the ranking of the remaining taxa severely.}
\label{fig:Tr66501}
\end{figure}

Note that TreeBase contains trees on a huge variety of species, not all of which are of interest for species conservation programs. However, our studies clearly show that taxon deletion from trees, which happens when species go extinct, can have dramatic effects on the FP index as a ranking criterion. As the comparison between the boxplots in Figure  \ref{fig:treebasestudy}  shows, this effect is generally larger the more species go extinct. We included only rooted trees in our studies which were binary and for which all branch lengths were given, but we suspect that similar effects can be seen for non-binary trees and trees with partial branch lengths as well.

\subsection{Simulations}\label{subsec:simulations}
Proposition \ref{Prop_caterpillar_ultrametric} motivated us to analyze the impact of branch lengths in more depth. It is well-known that the so-called Yule or Yule-Harding model (\citet{yulehartigan}), a pure birth model, leads to ultrametric (\enquote{clocklike}) trees. Proposition \ref{Prop_caterpillar_ultrametric} shows that at least ultrametric caterpillar trees cannot suffer from FP  rank swaps when leaves go extinct, which might suggest that having a low death rate (in the Yule model, the death rate is 0) prevents this problem. However, it turns out that actually the opposite is correct, as the following simulations show.

We used the computer algebra system Mathematica (\citet{mathematica}) to perform two studies. For these studies, we first simulated three sets of trees which were subsequently used for both studies. 

For all tree sets, we used $\lambda=10$ as a birth rate. For the first tree set, we used death rate $\mu=0$ (which corresponds to the Yule model and produces ultrametric trees), for the second tree set, we used $\mu=5$ and for the third tree set, we used $\mu=\lambda=10$. 

Each of the three tree sets contains 300 simulated trees, namely 100 trees each with 10 taxa, 30 taxa and 50 taxa, respectively. 

For all three tree sets, we then performed the same two studies as for the data set presented in the previous section, i.e. we first deleted the lowest ranked leaf from each tree and compared the resulting ranking with the ranking induced by the corresponding subtree of the original tree using Kendall's $\tau$. 
We then repeated this procedure, but deleted the lowest ranked leaf from each cherry (instead of only the overall lowest ranked leaf). 
The results of the two studies are presented in Figures \ref{fig:sim1} and  \ref{fig:sim2}. Note that both studies show the same overall trends as our analysis of the TreeBase data; namely that the more taxa a tree has, the lower the impact of a few rank swaps; and more leaf deletions tend to cause more rank swaps than a single leaf deletion.

However, our simulations show more than that: They show that the higher the death rate of branches in the simulated trees, i.e. the more non-ultrametric the tree is (and thus the more diverse the branch lengths are), the smaller the damage caused by leaf deletions. We suggest a possible explanation for this observation and discuss it more in-depth in the discussion section.

\begin{figure}
    \centering
    \includegraphics[scale=0.8]{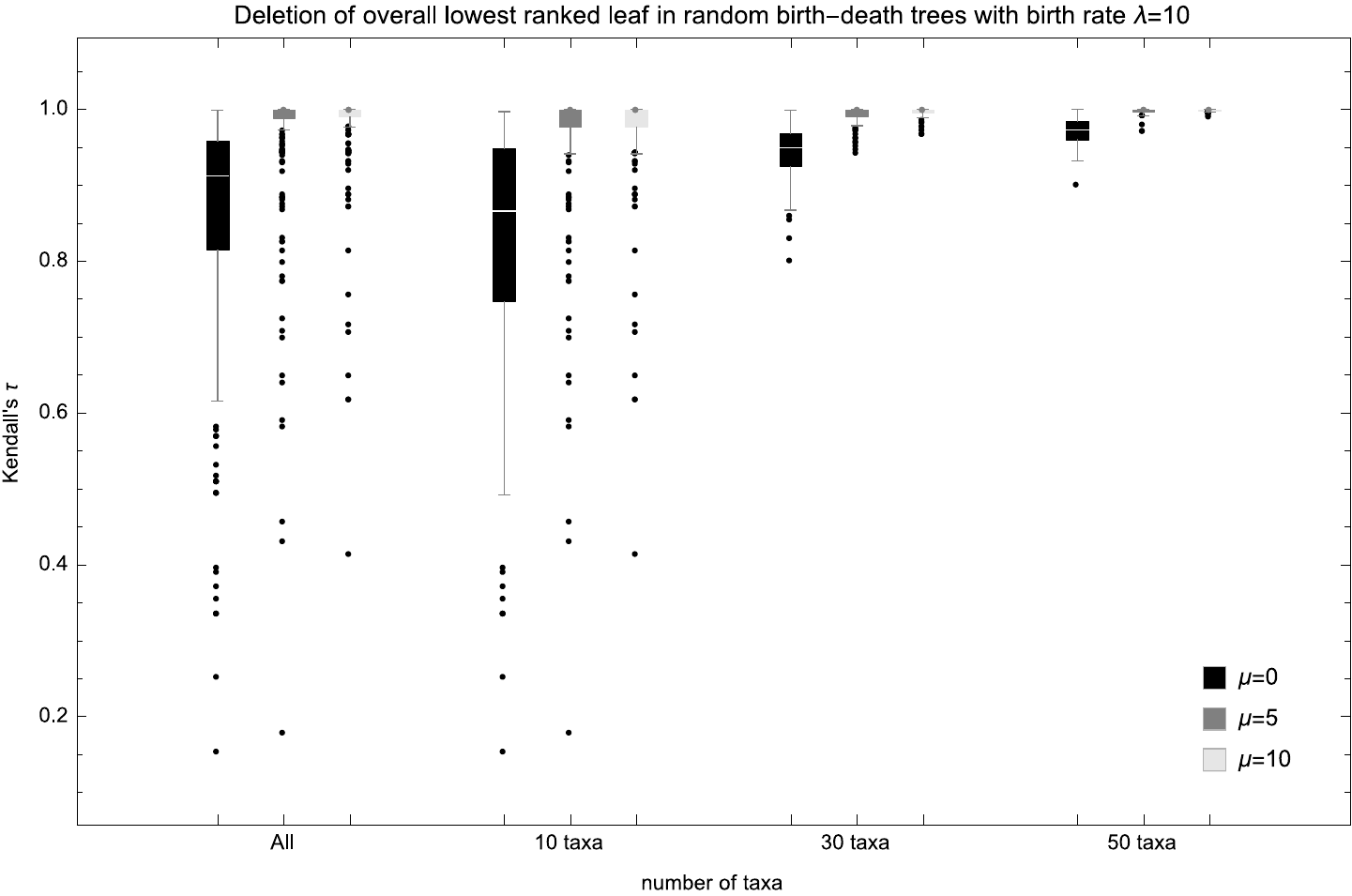}
    \caption{The resulting values of Kendall's $\tau$ when the leaf with the smallest FP value gets deleted and all FP values get re-calculated. The first boxplot contains all trees with up to 50 taxa that were simulated (with constant birth rate $\lambda=10$ and death rates $\mu=0$, $\mu=5$ and $\mu=10$, respectively), whereas the other boxplots are sorted by the numbers of taxa the respective trees contain.}
    \label{fig:sim1}
\end{figure}

\begin{figure}
    \centering
    \includegraphics[scale=0.8]{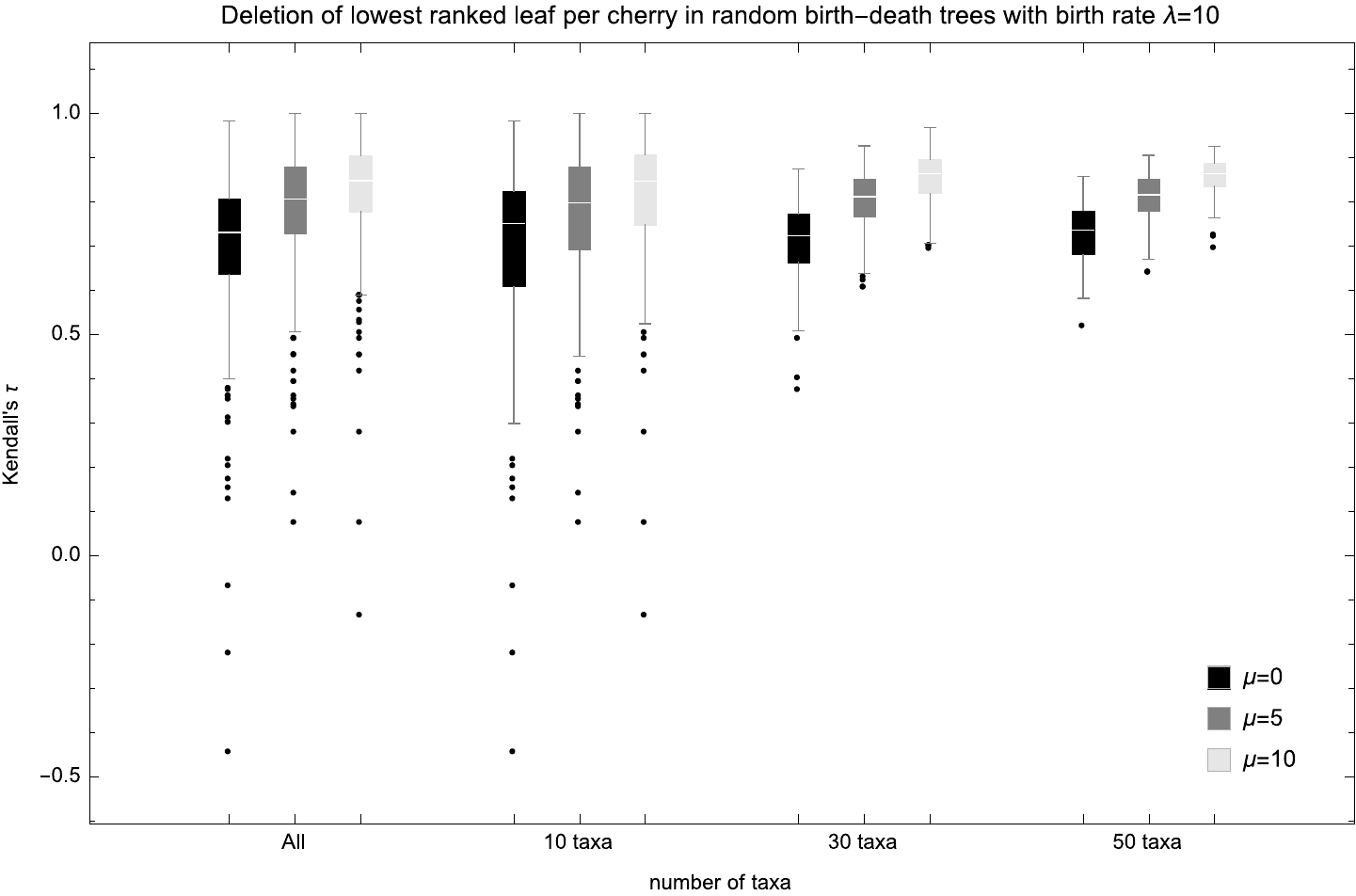}
    \caption{The resulting values of Kendall's $\tau$ when of each cherry the leaf with the smallest FP value gets deleted and all FP values get re-calculated. The first boxplot contains all trees with up to 50 taxa that were simulated (with constant birth rate $\lambda=10$ and death rates $\mu=0$, $\mu=5$ and $\mu=10$, respectively), whereas the other boxplots are sorted by the numbers of taxa the respective trees contain.}
    \label{fig:sim2}
\end{figure}

\section{Discussion}
\label{Sec_Discussion} 
The Fair Proportion index is a popular phylogenetic diversity index used to prioritize species for conservation. However, even if a species receives active 
conservation attention, there is still a risk that it goes extinct. The aim of the present manuscript was thus to analyze the effects of species extinction on the prioritization order obtained from the FP index. More specifically, we analyzed the extent to which the ranking order may change when some species go extinct and the FP index is re-computed for the remaining taxa. On the one hand, we showed that the extinction of one leaf per cherry might completely reverse the ranking. On the other hand, we proved that even the extinction of only the lowest ranked species in a tree can cause significant changes in the prioritization order. On the positive side, in the case of a single species extinction not involving the highest ranked species, we also showed that the number of species that require more urgent conservation attention than the formerly most important species in the remaining tree can at least be bounded from above. Here, we saw that the effects of a single species extinction are less dramatic if the underlying tree is ``balanced'' in the sense that its two maximal pending subtrees are of similar sizes, whereas the impact is more severe if these subtree sizes are very different. 
Note that the balance of a tree also played a major role when we showed that the extinction of one leaf per cherry can completely reverse the FP ranking as the number of cherries is in fact also often used to measure the balance of a tree (cf. \cite{McKenzie2000,Kersting2021}).
Investigating the impact of the shape of a tree, or more precisely its balance, on the FP index when species go extinct is thus an interesting direction for future research.

Moreover, the present results rely on a particular choice of edge lengths and do not immediately carry over to situations where restrictions on the edge lengths are in place. While we showed in Proposition \ref{Prop_caterpillar_ultrametric} that the FP ranking on ultrametric caterpillar trees is not affected by species extinction, we also saw that this is not the case for general ultrametric trees. On the contrary, our simulation results obtained subsequently showed that the more ultrametric a tree is, the more sensitive to leaf deletions it tends to be.
This can probably be explained by the fact that the FP indices of leaves in trees with high $\mu$-value (i.e. a high death rate) seem to have a higher variation than in ultrametric trees: In our simulations, the median variance for birth-death trees with $\lambda=10$ and $\mu=0$ was 0.00139859, while the median variance for birth-death trees with the same birth rate $\lambda$ but death rates $\mu=5$ and $\mu=10$ were higher, namely 0.00176227 and 0.00183703, respectively. This higher variation of FP indices amongst the leaves of these trees might imply that not so many of these leaves can simply swap their ranks because of extinctions. If there is less variation in the FP indices, i.e., the FP indices are more similar and closer together, their ranks can more easily be swapped when relatively few changes in the tree occur.
A second immediate direction for future research would thus be to analyze the effects of species extinction on the prioritization order obtained from the FP index when \emph{ultrametric} trees are considered more in-depth.

Another interesting direction for future research would be to analyze how other phylogenetic diversity indices, for example the so-called ``Equal-Splits index'' (\citet{Redding2003, Redding2006}), or prioritization indices based on other aspects of biodiversity such as ``feature diversity'' or ``functional diversity'' are affected by species extinctions and whether they are more ``robust'' than the FP index. 

We remark, however, that while we showed that the prioritization order obtained from the FP index might radically change when species go extinct, we do not suggest to disregard the FP index or other phylogenetic diversity indices completely. Our aim was merely to draw attention to these potential ``conservation regrets'' (i.e. cases where the initial choice of conservation priority might need substantial readjustment after events of species extinction). A way forward in this regard might be to perform a sensitivity analysis prior to conservation decisions to assess the impact of extinction events, e.g. of the lowest ranked species, on the prioritization order, and to adjust conservation attention accordingly. Developing a quantitative measure or test to assess the ``robustness'' of a phylogenetic diversity ranking for different scenarios (e.g., for the extinction of the lowest ranked species or the extinction of a highly ranked species despite conservation efforts) is thus another important direction for future research.

\section{Acknowledgements}
MF was supported by the joint research project DIG-IT!
funded by the European Social Fund (ESF), reference: ESF/14-BM-A55-
0017/19, and the Ministry of Education, Science and Culture of Mecklenburg-Vorpommerania, Germany. 
KW was supported by The Ohio State University's President's Postdoctoral Scholars Program. 
Last but not least, the authors wish to thank Sophie Kersting for checking the nexus files containing the simulated data for syntax correctness.

\section{Data and Supplementary Material}
The data underlying this article are available at: \url{http://mareikefischer.de/SupplementaryMaterial/FP_Index.zip}

\bibliographystyle{abbrvnat}
\bibliography{references.bib}

\section*{Appendix: Mathematical proofs}

\setcounter{thm}{0}
\begin{thm}
Let $T$ be a phylogenetic $X$-tree with $|X|\geq 3$. Let $ \pi_T$ be a strict and reversible ranking for $T$ with respect to $X'\subset X$ and induced subtree $\widetilde{T}$ on taxon set $\widetilde{X}=X\setminus X'$. Let $c=[x_i,x_j]$ be a cherry of $T$. Then, $X'$ contains at least one of the elements $x_i$, $x_j$, i.e. $|X'\cap \{x_i,x_j\}|\geq 1$. 
\end{thm}

\begin{proof} We prove the statement by contradiction. Therefore, let $T$ be a phylogenetic $X$-tree, and let $ \pi_T$ be a strict and reversible ranking for $T$ with respect to $X'\subset X$ and induced subtree $\widetilde{T}$ on taxon set $\widetilde{X}$. Moreover, let $c=[x_i,x_j]$ be a cherry of $T$ and assume that $\{x_i,x_j\} \cap X'=\emptyset$, i.e. $x_i, x_j \not\in X'$. Note that by definition, we have $$ FP_T(x_i)=\sum\limits_{e \in P(T; \rho,x_i)}\frac{\lambda_e}{D_e}$$ and, analogously, 
$$ FP_T(x_j)=\sum\limits_{e \in P(T;\rho,x_j)}\frac{\lambda_e}{D_e}.$$ 

Now note that as $x_i$ and $x_j$ form a cherry with some parent $w$, the paths $P(T;\rho,x_i)$ and $P(T;\rho,x_j)$ are identical except for the last edges, i.e. except for the pendant edges, say $e_i=(w,x_i)$ and $e_j=(w,x_j)$. This immediately implies $$ FP_T(x_i)=FP_T(x_j)-\frac{\lambda_{e_j}}{D_{e_j}}+\frac{\lambda_{e_i}}{D_{e_i}}=FP_T(x_j)-\lambda_{e_j}+\lambda_{e_i},$$  where the last equality is due to the fact that only $x_i$ descends from $e_i$ and only $x_j$ descends from $e_j$, and thus $D_{e_i}=D_{e_j}=1$. Analogously, we have $$ FP_{\widetilde{T}}(x_i)=FP_{\widetilde{T}}(x_j)-\widetilde{\lambda_{e_j}}+\widetilde{\lambda_{e_i}},$$ where $\widetilde{\lambda_{e_\cdot}}$  corresponds to the lengths of the respective pendant edges in $\widetilde{T}$. However, as the entire cherry $[x_i,x_j]$ is preserved in $\widetilde{T}$, we know that $\lambda_{e_i}=\widetilde{\lambda_{e_i}}$ and $\lambda_{e_j}=\widetilde{\lambda_{e_j}}$. This is due to the fact that as $e_i$ does not get deleted from $T$ to $\widetilde{T}$, $e_j$ does not get merged with the edge leading from the parent of $w$ to $w$, so both edges remain unchanged. Thus, in total we have 

\begin{equation}\label{eq_TtildeTdifference}FP_T(x_i)-FP_T(x_j)=\lambda_{e_i}-\lambda_{e_j}=FP_{\widetilde{T}}(x_i)-FP_{\widetilde{T}}(x_j).\end{equation}

Without loss of generality, we assume that $FP_T(x_i) > FP_T(x_j)$ in $\pi_T$ (as $\pi_T$ is strict), which implies that $FP_{\widetilde{T}}(x_j)>FP_{\widetilde{T}}(x_i)$ in $\pi_{\widetilde{T}}$ (as $\pi_T$  is reversible). However, this implies both $FP_T(x_i) - FP_T(x_j)>0$ and $FP_{\widetilde{T}}(x_i)-FP_{\widetilde{T}}(x_j)<0$, which implies   $FP_T(x_i) - FP_T(x_j) \neq FP_{\widetilde{T}}(x_i)-FP_{\widetilde{T}}(x_j)$ and thus contradicts Equation \eqref{eq_TtildeTdifference}. This completes the proof. 
\end{proof}


\begin{thm}
Let $n\geq 2$ and let $T=(T_a,T_b)$ be a rooted binary phylogenetic $X$-tree with $|X|=n$. 
Let $\widetilde{T}$ be the induced subtree on leaf set $\widetilde{X}\subset X$ that results from $T$ when we delete one leaf out of each cherry of $T$ and suppress the resulting vertices of in-degree 1 and out-degree 1. Then, we have:

There exist strictly positive edge lengths $\lambda_1,\ldots,\lambda_{2n-2}$ for $T$ such that there is a strict ranking $\pi_T$ for the leaves of $T$ concerning the FP index which is reversible with respect to $\widetilde{T}$ and such that $\widetilde{T}$ contains the species which has the highest FP index in $T$ and such that the species with the lowest FP index in $T$ is \emph{not} contained in $\widetilde{T}$. In particular, if $c_T$ denotes the number of cherries of $T$ and if $x_1 \in X$ is such that $x_1=\argmax\limits_{x \in X}FP_T(x)$, then we have $FP_T(x_1)>FP_T(x_2)>\ldots > FP_T(x_n)$ and $FP_{\widetilde{T}}(x_1)<FP_{\widetilde{T}}(\widetilde{x}_2)< \ldots < FP_{\widetilde{T}}(\widetilde{x}_{n-c_T})$, where $\widetilde{x}_i \in \widetilde{X}$ for all $i=2,\ldots,n-c_T$ and where $FP_T(\widetilde{x}_i)>FP_T(\widetilde{x}_j)$ if and only if $FP_{\widetilde{T}}(\widetilde{x}_i)<FP_{\widetilde{T}}(\widetilde{x}_j)$. Moreover, if $x' \in X$ is such that $x'=\argmin\limits_{x \in X}FP_T(x)$, then $x' \not\in \widetilde{X}$.
\end{thm}

In order to prove the above theorem, we need a few lemmas that state some properties of rankings and the FP index in general. We start by showing that the FP index cannot decrease for any species if leaves are deleted from a tree $T=(T_a,T_b)$ with $n \geq 2$ leaves (as long as neither all leaves of $T_a$ nor all leaves of $T_b$ are deleted).

\begin{lem} \label{lem_deletionincrease} Let $T=(T_a,T_b)$ be a rooted binary phylogenetic $X$-tree with maximal pendant subtrees $T_a$ and $T_b$ with taxon sets $X_a$ and $X_b$, respectively. Let $X'\subset X$ be such that $X_a,X_b \not\subseteq X'$, i.e. neither $X_a$ nor $X_b$ are completely contained in $X'$. Moreover, let $\widetilde{T}$ be the induced subtree on $\widetilde{X}=X\setminus X'$ resulting from $T$ when the taxa of $X'$ are deleted. Then, we have:  $FP_T(x) \leq FP_{\widetilde{T}}(x)$ for all $x \in \widetilde{X}$. 
\end{lem}

\begin{proof} Let $x \in \widetilde{X}$. Then, either the unique path from the root $\rho$ of $T$ to $x$ contains at least one edge $e$ that also occurs on one of the unique paths from $\rho$ to taxa in $X'$ or not. If not, then $x$ is not affected by the deletion of $X'$ at all, i.e. $FP_T(x)=FP_{\widetilde{T}}(x)$ (note that this would not be true if we allowed $X_a$ or $X_b$ to be completely contained in $X'$, because then the deletion of the entire corresponding subtree would enforce a deletion of the edge leading to the other subtree and thus have an impact on the FP indices of the remaining taxa). However, if such an edge $e$ exists, then a proportion of the edge length of $e$ is assigned to $x$ by the FP index. This proportion, however, increases when the taxa of $X'$ are deleted, because then fewer leaves are descended from $e$ and the FP index will account for this. As this holds for all such edges, in this case we have $FP_T(x) <FP_{\widetilde{T}}(x)$. Thus, altogether we have $FP_T(x) \leq FP_{\widetilde{T}}(x)$. This completes the proof.
\end{proof}

The previous lemma showed that leaf deletion can keep the FP indices of the remaining leaves unchanged or increase them, but never decrease them. The following lemma adds to this for the case that the given edge lengths and the leaf deletions induce a strict and reversible ranking: In this case, there is at most one taxon whose FP index remains unchanged, namely the one whose such value is maximal.

\begin{lem} \label{lem_onlyoneunchanged} Let $T$ be a rooted binary phylogenetic $X$-tree. Let $X'\subset X$ and let $\widetilde{T}$ be the induced subtree on $\widetilde{X}=X\setminus X'$ resulting from $T$ when the taxa of $X'$ are deleted. Assume that the edge lengths of $T$ are such that the deletion of $X'$ induces a strict and reversible ranking $\pi_T$. Let $x'=\argmax\limits_{x \in X} FP_T(x)$. Then, we have:  $FP_T(x) <  FP_{\widetilde{T}}(x)$ for all $x \in \widetilde{X}\setminus \{x'\}$ and, if $x' \in \widetilde{X}$,  $FP_T(x') \leq  FP_{\widetilde{T}}(x')$.
\end{lem}

\begin{proof} We know that $FP_T(x) \leq FP_{\widetilde{T}}(x) $ for all $x \in \widetilde{X}$ by Lemma \ref{lem_deletionincrease}. It only remains to show that for all $x \neq x'$ the inequality is strict. Assume that it is not, i.e. assume that there is a $y \in \widetilde{X} \setminus \{x'\}$ such that $FP_T(y) = FP_{\widetilde{T}}(y)$. As $y \neq x'$ and as $x'=\argmax\limits_{x \in X} FP_T(x)$ and as $\pi_T$ is strict and reversible, we know that $FP_T(y)<FP_T(x')$ and $FP_{\widetilde{T}}(y)>FP_{\widetilde{T}}(x')$. By Lemma \ref{lem_deletionincrease} we have $FP_T(x') \leq FP_{\widetilde{T}}(x')$. So in summary, this gives  $FP_T(y)<FP_T(x') \leq FP_{\widetilde{T}}(x') < FP_{\widetilde{T}}(y)=FP_T(y)$. The latter equality is due to our assumption. So in summary, we have $FP_T(y)<FP_T(y)$. Clearly, this is a contradiction and therefore the assumption was wrong. This completes the proof.
\end{proof}

The next two simple lemmas provide important properties of the FP index: They basically show us how we can modify the edge lengths of a given tree $T$ such that the rankings of the taxa obtained from the FP indices are not changed. These two lemmas are the crucial tools which we later need to prove Theorem \ref{thm_cherries}, because they allow us to scale branch lengths and enlarge pending branches of a tree and still keep a ranking strict and reversible. This will be summarized by Corollary \ref{cor_lem3lem4}. We start with the following lemma that allows us to upscale or downscale an entire tree by multiplying its edge lengths with a constant.

\begin{lem} \label{lem_FPscale} Let $T$ be a rooted binary phylogenetic $X$-tree with $|X|=n$ and edge lengths $\lambda_1,\lambda_2\ldots, \lambda_{2n-2}$. Let $T'$ be like $T$, but with the edge lengths multiplied by $k$ for some $k \in \mathbb{R}_+$, i.e. $T'$ has edge lengths  $\lambda_1\cdot k,\lambda_2\cdot k\ldots, \lambda_{2n-2} \cdot k$. Then, we have: $FP_{T'}(x)=FP_T(x)\cdot k$ for all $x \in X$. In particular, we have $FP_T(x)>FP_{T}(y) \Leftrightarrow FP_{T'}(x)>FP_{T'}(y)$ for $x,y \in X$. 
\end{lem}

\begin{proof} Let $x\in X$. Then, by definition of $FP$, we have $FP_T(x)=\sum\limits_{e \in P(T; \rho,x)}\frac{\lambda_e}{D_e}$ and  $FP_{T'}(x)=\sum\limits_{e \in P(T';\rho,x)}\frac{k\cdot \lambda_e}{D_e}=k\cdot \sum\limits_{e \in P(T;\rho,x)}\frac{\lambda_e}{D_e} = k \cdot FP_T(x)$. For $k \in \mathbb{R}_+$, this immediately implies $FP_T(x)>FP_{T}(y) \Leftrightarrow FP_{T'}(x)=k\cdot FP_T(x)>k\cdot FP_{T}(y)=FP_{T'}(y)$. This completes the proof.
\end{proof}

The following lemma shows that when we modify the edge lengths of all pendant edges of a tree $T$ by adding a positive constant to them, this cannot change the ordering of the $FP_T$ values.

\begin{lem} \label{lem_FPconst} Let $T$ be a rooted binary phylogenetic $X$-tree with $|X|=n$ and edge lengths $\lambda_T(e)$ for $e \in E(T)$. Let $T'$ be like $T$, but such that the edge lengths of all pendant edges are longer by a constant $d \in \mathbb{R}_+$. More precisely, for all edge lengths $\lambda_{T'}(e)$ of edges $e$ of $T'$, let $$ \lambda_{T'}(e)=\begin{cases}  \lambda_{T}(e)+d & \mbox{if $e$ is a pendant edge of $T$ and $T'$, }\\  \lambda_{T}(e) & \mbox{else.}\end{cases}$$ Then, we have $FP_{T'}(x)=FP_T(x)+d$ and $FP_T(x)>FP_T(y) \Leftrightarrow FP_{T'}(x) > FP_{T'}(y)$ for all $x,y \in X$. 
\end{lem}

\begin{proof} Let $x\in X$ and $d \in \mathbb{R}_+$. Let $p_x$ denote the parent of $x$ and let $e_x=(p_x,x)$ be the pendant edge incident to $x$. Then, by definition of $FP$, we have $FP_T(x)=\sum\limits_{e \in P(T; \rho,x)}\frac{\lambda_T(e)}{D_e}$ and  $FP_{T'}(x)=\sum\limits_{e \in P(T';\rho,x)}\frac{ \lambda_{T'}(e)}{D_e}= \sum\limits_{e \in P(T';\rho,p_x)}\frac{ \lambda_{T'}(e)}{D_e}+\frac{ \lambda_{T'}(e_x)}{1}=\sum\limits_{e \in P(T;\rho,p_x)}\frac{ \lambda_{T}(e)}{D_e}+ \lambda_{T}(e_x)+d=FP_T(x)+d$. This completes the first part of the proof. Now let $x,y \in X$ and $d \in \mathbb{R}_+$. Then we have: $FP_T(x)>FP_T(y) \Leftrightarrow FP_{T'}(x)=FP_T(x)+d > FP_T(y)+d=FP_{T'}(y)$. This completes the proof.
\end{proof}

We are now in a position to state the following simple but crucial corollary, which is the main ingredient in the inductive steps of the proofs of Theorems \ref{thm_cherries} and \ref{thm_boundoneleaf}. 

\begin{cor} \label{cor_lem3lem4}
Let $T$ be a rooted binary phylogenetic $X$-tree with $|X|=n$. Let $\lambda_T(e)$ for $e \in E(T)$ be an edge length assignment for $T$ that 
induces a strict and reversible ranking $\pi_T$ concerning the deletion of some taxa $X' \subset X$. Let $d, k \in \mathbb{R}_+$. Then, if we multiply all edge lengths by $k$ and subsequently add $d$ to pendant edge lengths, the resulting branch lengths induce a strict and reversible ranking conerning the deletion of $X'$, too.
\end{cor}

\begin{proof}
By Lemmas \ref{lem_FPscale} and \ref{lem_FPconst}, the ranking $\pi_T$ of the taxa in $X$ induced by $FP_T$ is not affected by scaling all branch lengths by $k$ or adding $d$ to the pendant edges. So we only need to show that the same is true for the ranking $\pi_{\widetilde{T}}$ of the taxa in $\widetilde{X}:=X\setminus X'$. However, note that in $\widetilde{T}$ (i.e. in the subtree remaining when the taxa in $X'$ get deleted), due to the suppression of degree-2 vertices, some edges get merged, but as both have been scaled by the same factor $k$, this also holds for the new long edge. For instance, if edges $e_1$ and $e_2$ get merged, the scaling leads to two edges $k\cdot \lambda_{e_1}$ and $k\cdot \lambda_{e_2}$ in $T$, but to a single edge of length $k\cdot \lambda_{e_1}+k\cdot \lambda_{e_2}=k\cdot (\lambda_{e_1}+\lambda_{e_2})$ in $ \widetilde{T}$. So the edges in $\widetilde{T}$ get scaled by the same factor as the edges in $T$. In particular, by Lemma \ref{lem_FPscale}, this does not affect their ranking. Moreover, if we subsequently add $d$ to the pendant edges of $T$, this also adds $d$ to the pendant edges of $\widetilde{T}$ (regardless of whether these edges are merged with some inner edges or not). Thus, by Lemma \ref{lem_FPconst}, this again does not affect the ranking induced by $FP_{\widetilde{T}}$, which completes the proof.
\end{proof}

The final lemma that we need to prove Theorem \ref{thm_cherries} is somewhat technical. It shows that if you delete precisely one leaf per cherry in a rooted binary phylogenetic tree $T$ with more than two leaves, the remaining tree has at least two leaves.

\begin{lem}\label{lem_atleast2leaves} Let $T$ be a rooted binary phylogenetic $X$-tree with $|X|=n>2$. Let $\widetilde{T}$ be the tree resulting from $T$ when precisely one leaf from each cherry of $T$ is deleted (and the resulting vertex with in-degree 1 and out-degree 1 is suppressed, respectively). Then, $\widetilde{T}$ has at least two leaves. \end{lem}

\begin{proof} Let $c_T$ denote the number of cherries in $T$. Clearly, $c_T \leq \lfloor\frac{n}{2} \rfloor$, because if $n$ is even, at most \emph{all} leaves can be contained in a cherry, in which case there are $\frac{n}{2}$ cherries, and if $n$ is odd, at most $n-1$ leaves can be contained in cherries, in which case there are $\frac{n-1}{2}$ cherries. So if we now delete one leaf per cherry, this implies we delete $c_T\leq \lfloor\frac{n}{2} \rfloor$ leaves. Thus, for the number $n_{\widetilde{T}}$ of leaves in $\widetilde{T}$ we have: $n_{\widetilde{T}} = n-c_T \geq n- \lfloor\frac{n}{2} \rfloor \geq n- \frac{n}{2} =  \frac{n}{2}>1.$ The latter inequality is due to the fact that $n>2$ by assumption. So we have that $\widetilde{T}$ has more than one leaf; so it must have at least two leaves. This completes the proof. 
\end{proof}

We are now finally in the position to prove the main theorem of this section, namely Theorem  \ref{thm_cherries}.

\begin{proof}[Proof of Theorem \ref{thm_cherries}]
We prove the statement by induction on $n$. If $n=2$, there is only one rooted binary tree shape, namely the one that consists only of a cherry, say $[x_1,x_2]$ and pendant edges $e_1$ and $e_2$ with edge lengths $\lambda_1$ and $\lambda_2$, respectively. Then, $FP_T(x_1)=\lambda_1$ and $FP_T(x_2)=\lambda_2$. We choose $\lambda_1>\lambda_2>0$. So if we now delete leaf $x_2$, which has minimal $FP_T$ value, we derive a tree $\widetilde{T}$ that only has one leaf $x_1$. So $x_1$, the leaf that formerly had maximum $FP_T$ value, is still present, and -- because it is the only leaf -- it has minimum $FP_{\widetilde{T}}$ value. This completes the base case of the induction. 

Now assume the statement holds for all trees with at most $n-1$ leaves and let $T=(T_a,T_b)$ be a tree with $n$ leaves. We now distinguish between two cases:

\begin{enumerate}[(i)]
\item $T_b$ consists of only one leaf, i.e. $n_b=1$. Without loss of generality, we may assume that leaf $x_n$ is the unique vertex in $T_b$. In this case, $T$ looks as depicted in Figure \ref{fig_T_in_part1ofcherryproof}. Note that in this case, all cherries of $T$ are actually contained in $T_a$. Then by the inductive hypothesis, as $T_a$ has strictly fewer leaves than $T$, there exist edge lengths for $T_a$ that induce a strict and reversible ranking $\pi_{T_a}$ which fulfills all requirements stated by the theorem. We fix these edge lengths accordingly and calculate $FP_{T_a}(x_1),\ldots, FP_{T_a}(x_{n-1})$. Without loss of generality, we have \begin{equation}\label{FPTa} FP_{T_a}(x_1)> FP_{T_a}(x_2)> \ldots > FP_{T_a}(x_{n-1}).\end{equation} In particular, we may assume $FP_{T_a}(x_1)=\max\limits_{x \in X_a}\{FP_{T_a}(x)\}$ (otherwise, we could relabel the taxa of $T$ accordingly). As $\pi_{T_a}$ is reversible by assumption and as $\widetilde{T_a}$ contains taxon $x_1=\argmax\limits_{x \in X_a}FP_{T_a}(x)$, this immediately leads to \begin{equation} \label{FPTaTatilde} FP_{\widetilde{T_a}}(x_1)< FP_{\widetilde{T_a}}(\widetilde{x}_2)< \ldots < FP_{\widetilde{T_a}}(\widetilde{x}_{n-1-c_T}),\end{equation}

where $c_T$ denotes the number of cherries in $T$ (and $T_a$) and thus the number of deleted leaves. Moreover, we know by assumption that$x_1 \in \widetilde{X_a}$, and we call the other elements of   $\widetilde{X_a}$ $\widetilde{x}_2,\ldots, \widetilde{x}_{n-1-c_T}$.

\begin{figure}
	\centering
	\includegraphics[scale=0.3]{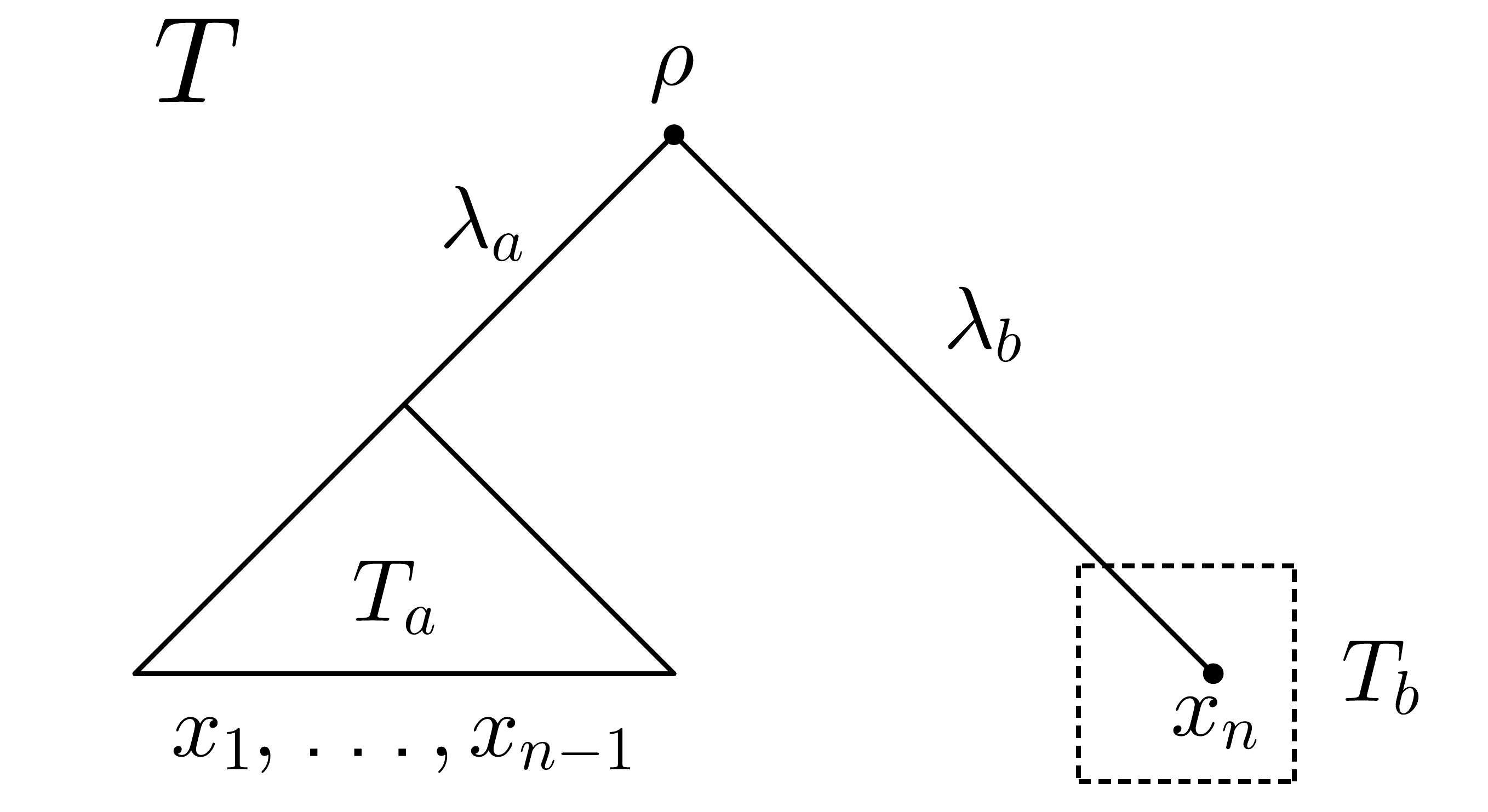}
	\caption{Tree $T$ in part 1 of the proof of Theorem \ref{thm_cherries}. In this case, $T_b$ consists of only one vertex, and all cherries of $T$ are contained in $T_a$.}
	\label{fig_T_in_part1ofcherryproof}
\end{figure}

 Note that we know (as ranking $\pi_{T_a}$ is reversible) that $\widetilde{x}_i < \widetilde{x}_j$ if and only if $i <j$. This means that the subset that remains when one leaf per cherry is deleted from $T_a$ equals $\widetilde{X_a}=\{x_1,\widetilde{x}_2,\ldots,\widetilde{x}_{n-1-c_T}\}\subset X_a$ and that the elements of $ \widetilde{X_a}$ are strictly reversed by $FP_{\widetilde{T_a}}$  compared to $FP_{T_a}$. 

We now use ranking $\pi_{T_a}$ to construct the desired ranking $\pi_T$. Therefore, note that $FP_T(x)=FP_{T_a}(x)+\frac{\lambda_a}{n-1}$ for all $x=x_1,\ldots,x_{n-1}$. Using \eqref{FPTa}, this immediately leads to \begin{equation}\label{FPT} FP_{T}(x_1)> FP_{T}(x_2)> \ldots > FP_{T}(x_{n-1}).\end{equation}

Note that we also know that $\widetilde{X_a}$ does \emph{not} contain taxon $x_{n-1}$ as $\pi_{T_a}$ fulfills all requirements of the theorem by the inductive hypothesis, so $x_{n-1}=\argmin\limits_{x\in X_a} FP_{T_a}(x)$ has to be affected by the cherry leaf deletion.

Next, by Lemma \ref{lem_deletionincrease}, the deletion of one leaf per cherry can only further increase the FP index of a taxon in $T_a$ or not change it at all, but it cannot decrease it, so we have $FP_{\widetilde{T_a}}(x)\geq FP_{T_a}(x)$ for all $x \in \widetilde{X_a}$. Moreover, none of these deletions can affect $T_b$, as this subtree consists only of one leaf, which on its path to the root of $T$ does not share an edge with any of the leaves of $T_a$. So we have \begin{equation}\label{FPafterdeletion} FP_{\widetilde{T_a}}(x)+\frac{\lambda_a}{n-1-c_T}=FP_{\widetilde{T}}(x)> FP_T(x)=FP_{T_a}(x)+\frac{\lambda_a}{n-1} \mbox{ for all } x\in \widetilde{X_a}.\end{equation} Here, the strict inequality is due to the fact that $FP_{\widetilde{T_a}}(x)\geq FP_{T_a}(x)$ and $\frac{\lambda_a}{n-1-c_T}>\frac{\lambda_a}{n-1}$ for all $x\in \widetilde{X_a}$, as $c_T\geq 1$. Moreover, we have $FP_{\widetilde{T}}(x_n)=FP_T(x_n)=\lambda_b$.

Combining Equations \eqref{FPTaTatilde} and \eqref{FPafterdeletion}, we immediately get \begin{equation}\label{comb} FP_{\widetilde{T}}(x_1)< FP_{\widetilde{T}}(\widetilde{x}_2)< \ldots < FP_{\widetilde{T}}(\widetilde{x}_{n-1-c_T}).\end{equation}

We now choose $\lambda_a>0$ arbitrarily, e.g.  $\lambda_a=1$, and then want to choose $\lambda_b$ such that $FP_T(x_n)=\lambda_b>FP_T(x_1)$ and $FP_{\widetilde{T}}(x_n)=\lambda_b<FP_{\widetilde{T}}(x_1)$. By Equation  \eqref{FPafterdeletion}, we have $FP_T(x_1)<FP_{\widetilde{T}}(x_1)$, so we can simply choose $\lambda_b$ between these two values. For instance, we can set $\lambda_b:=FP_{T_a}(x_1)+\frac{1}{2}\left(\frac{\lambda_a}{n-1}+\frac{\lambda_a}{n-1-c_T} \right)$. This choice of $\lambda_b$ is in the middle between $FP_T(x_1)$ and $FP_{\widetilde{T}}(x_1)$. Therefore, using Equation \eqref{FPT} we now have:

\begin{equation} \lambda_b=FP_T(x_n)>FP_{T}(x_1)> FP_{T}(x_2)> \ldots > FP_{T}(x_{n-1}). \end{equation}

Moreover, using Equation \eqref{comb}, we immediately get: \begin{equation} \lambda_b=FP_{\widetilde{T}}(x_n)<FP_{\widetilde{T}}(x_1)< FP_{\widetilde{T}}(\widetilde{x}_2)< \ldots < FP_{\widetilde{T}}(\widetilde{x}_{n-1-c_T}). \end{equation}

Therefore, we have found a strict and reversible ranking for $T$ with respect to $\widetilde{T}$. Moreover, we know that taxon $x_n$ has maximal FP index in $T$, i.e. $x_n=\argmax\limits_{x \in X} FP_T(x)$, and $x_n$ is contained in $\widetilde{X}$. Together with the fact that taxon $x_{n-1}$, which was deleted from $X_a$ to get $\widetilde{X_a}$ and which has the lowest $FP_{T_a}$ value by \eqref{FPTa} and thus also the lowest $FP_T$ value by \eqref{FPT}, is \emph{not} contained in $\widetilde{X}$, this completes the first part of the proof.

\item We now consider the case where both $T_a$ and $T_b$ contain at least two leaves each, i.e. $|X_a|=n_a\geq 2$ and $|X_b|=n_b\geq 2$ and $n_a\geq n_b$. In particular, this implies that both $T_a$ and $T_b$ have at least one cherry and are thus affected by the described leaf deletion (one leaf per cherry). In the following, let $c_a$ and $c_b$ denote the numbers of cherries of $T_a$ and $T_b$, respectively. As $n_a$ and $n_b$ are both strictly smaller than $n$, for $T_a$ and $T_b$ we know by induction that there are edge lengths that allow for strict and reversible rankings $\pi_{T_a}$ and $\pi_{T_b}$ that fulfill all requirements stated by the theorem. We use these rankings to construct a strict and reversible ranking $\pi_T$ for $T$.

We first consider the case where $n_a=2$. As $2=n_a\geq n_b\geq 2$, this immediately implies $n_b=2$. Then, $T$ look as depicted in Figure \ref{fig_4bal}, and we choose the edge lengths as given in that Figure. The caption of this figure explains why the depicted edge lengths lead to a strict and reversible ranking which fulfills all requirements of the theorem. So for $n_a=2$, there remains nothing to show. 

Thus, we may assume from now on that $n_a>2$. By Lemma \ref{lem_atleast2leaves}, we may conclude that when we delete one leaf from each cherry of $T_a$ to obtain $\widetilde{T_a}$, $\widetilde{T_a}$ still has at least two leaves, say $x_1$ and $x_2$. 

\begin{figure}
	\centering
	\includegraphics[scale=0.3]{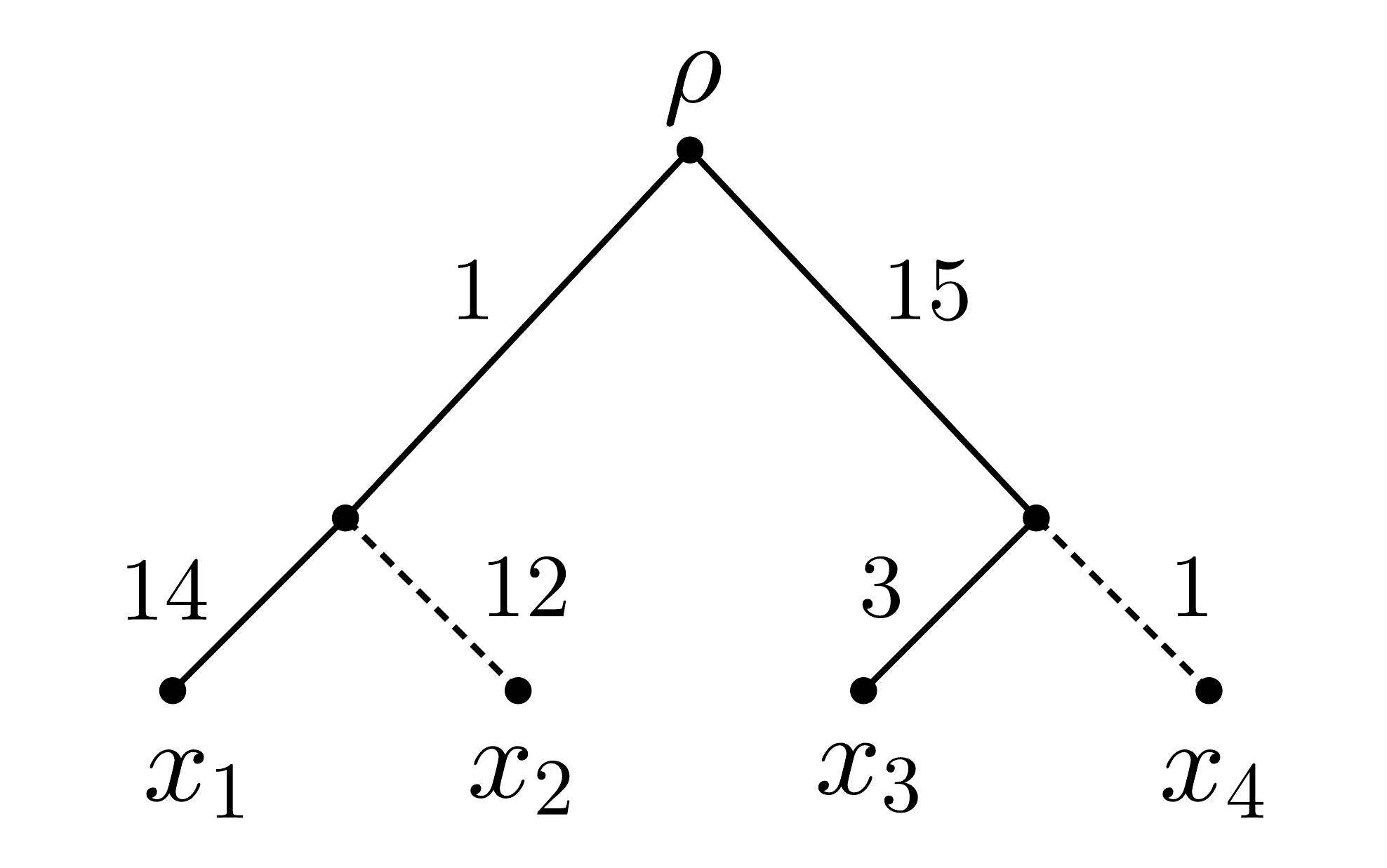}
	\caption{Tree $T$ with the depicted edge lengths leads to $FP_T(x_1)= 14.5   > FP_T(x_2) =12.5  > FP_T(x_3)=  10.5 > FP_T(x_4)=8.5 $. We then construct $\widetilde{T}$ by deleting leaves $x_2$ and $x_4$. This leads to $FP_{\widetilde{T}}(x_1)=15 < 18 = FP_{\widetilde{T}}(x_3)$. This shows that the ranking $\pi_T$ induced by the depicted edge lengths is strict and reversible. Note that leaf $x_1$ with the highest $FP_T$ value is still present after the deletion, while leaf $x_4$ with the lowest $FP_T$ value is amongst the deleted ones.}
	\label{fig_4bal}
\end{figure}

We now consider the rankings $\pi_{T_a}$ and $\pi_{T_b}$ and slightly modify the edge lengths for $T_a$ and $T_b$, but such that the resulting induced rankings remain in the exact same order as suggested by $\pi_{T_a}$ and $\pi_{T_b}$:

\begin{itemize}
\item We multiply all edges of $T_b$ by a factor $s_b>0$. The resulting edge lengths still induce a strict and reversible ranking by Corollary \ref{cor_lem3lem4}, and this ranking has the same properties as $\pi_{T_b}$ in the sense that it still fulfills all requirements of the theorem.
\item We then add a constant $d_b>0$ to all pendant edges of $T_b$. The resulting edge lengths still induce a strict and reversible ranking by Corollary \ref{cor_lem3lem4}. Again, this ranking has the same properties as $\pi_{T_b}$ in the sense that it still fulfills all requirements of the theorem.
\item We add a constant $d_a>0$ to all pendant edges of $T_a$. Again, the resulting edge lengths still induce a strict and reversible ranking by Corollary \ref{cor_lem3lem4} and still fulfill all requirements of the theorem. 
\end{itemize}

Let $\lambda_a$ denote the length of the edge leading from the root of $T$ to $T_a$, and let $\lambda_b$ denote the length of the edge leading from the root of $T$ to $T_b$. Then, $T$ now looks as depicted in Figure \ref{fig_modifiedTforproof}, and the FP indices of the leaves of $T$ are as follows:

\begin{equation} \label{eq_Ta} FP_T(x)= FP_{T_a}(x)+d_a+\frac{\lambda_a}{n_a} \mbox{ for all } x \in X_a , \mbox{ and }\end{equation}

\begin{equation} \label{eq_limit} FP_T(x)= FP_{T_b}(x)\cdot s_b+d_b+\frac{\lambda_b}{n_b}  \mbox{ for all } x \in X_b.\end{equation}

\begin{figure}
	\centering
	\includegraphics[scale=0.3]{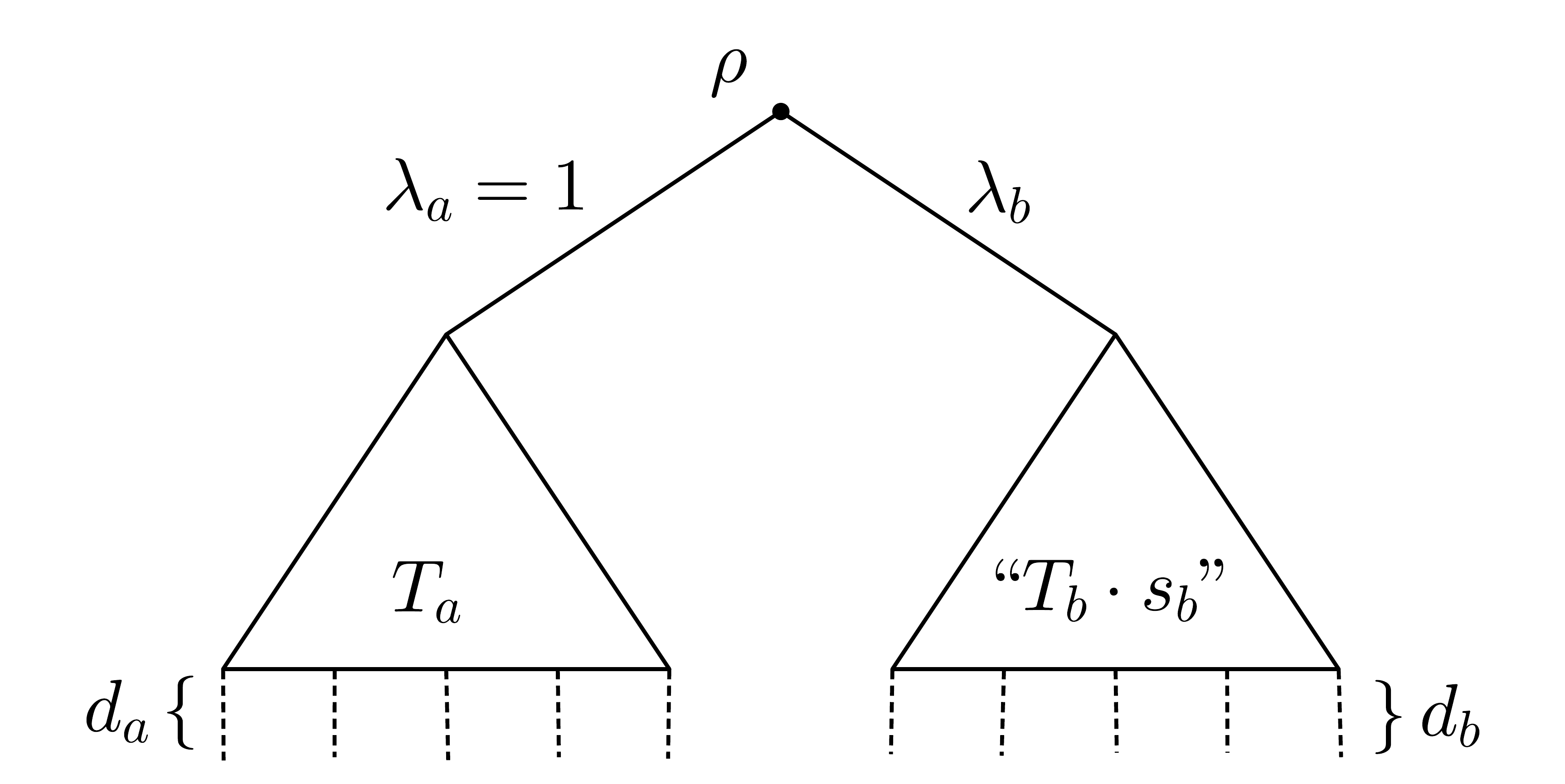}
	\caption{Tree $T$ with edge lengths based on the edge lengths of $T_a$ and $T_b$ that allow for a strict and reversible ranking, but with the modifications described in the second part of the inductive step in the proof of Theorem \ref{thm_cherries}:  Subtree $T_b$ is scaled by a factor $s_b>0$ (i.e. all edge lengths in $T_b$ are multiplied by $s_b$). Then, all pendant edges in $T_b$ are increased by a constant $d_b>0$. Last, all pendant edges in $T_a$ are increased by a constant $d_a>0$. 
	}
	\label{fig_modifiedTforproof}
\end{figure}

Now we delete all leaves from cherries of $T$ that had to be deleted from $T_a$ and $T_b$ to get to the strict and reversible rankings $\pi_{T_a}$ and $\pi_{T_b}$ we started with. Assume this implies that we delete $c_a$ cherry leaves from $T_a$ and $c_b$ cherry leaves from $T_b$. This results in tree  $\widetilde{T}=(\widetilde{T_a},\widetilde{T_b})$ with leaf set $\widetilde{X}$, where the leaf sets of $\widetilde{T_a}$ and $ \widetilde{T_b}$ are denoted by $\widetilde{X_a}$ and $\widetilde{X_b}$, respectively. 

Similarly as above, we can now express $FP_{\widetilde{T}}(x)$ for all taxa $x \in \widetilde{X}$: 

\begin{equation} \label{eq_limTatilde} FP_{\widetilde{T}}(x)= FP_{\widetilde{T_a}}(x)+d_a+\frac{\lambda_a}{n_a-c_a} \mbox{ for all } x \in \widetilde{X_a} , \mbox{ and }\end{equation}

\begin{equation} \label{eq_limit2} FP_{\widetilde{T}}(x)= FP_{\widetilde{T_b}}(x)\cdot s_b+d_b+\frac{\lambda_b}{n_b-c_b} \mbox{ for all } x \in \widetilde{X_b}.\end{equation}

Now, by Equations \eqref{eq_limit} and \eqref{eq_limit2} it is clear that as $s_b \rightarrow 0$, we have $FP_T(x) \rightarrow d_b+\frac{\lambda_b}{n_b}$ and $FP_{\widetilde{T}}(x)\rightarrow d_b+\frac{\lambda_b}{n_b-c_b}$ for all $x \in X_b$. 

Now, as stated above, we may assume that $\widetilde{T_a}$ has at least two leaves. So now let $x_1$, $x_2$ be two leaves from $T_a$ which are also both present in $\widetilde{T_a}$ and for which we have $FP_{\widetilde{T_a}}(x_1)<FP_{\widetilde{T_a}}(x_2)$ such that  there is \emph{no} taxon $y$ in $\widetilde{T_a}$ such that $FP_{\widetilde{T_a}}(x_1)<FP_{\widetilde{T_a}}(y)< FP_{\widetilde{T_a}}(x_2)$ (i.e. $x_1$ and $x_2$ are direct neighbors in the ranking induced by $FP_{\widetilde{T}}$). Moreover, as $\pi_{T_a}$ is strict and reversible by the inductive hypothesis, we have $FP_{T_a}(x_1)>FP_{T_a}(x_2)$. 

The high-level idea now is to perform the following two steps:

\begin{itemize}
\item First, we show that we can choose suitable values for $d_a$, $d_b$, $\lambda_a$ and $\lambda_b$ such that the limits of $FP_T(x)$ and $FP_{\widetilde{T}}(x)$ for $s_b \rightarrow 0$ lie in the open intervals $\left( FP_{T}(x_2), FP_{T}(x_1)\right)$ and $\left( FP_{\widetilde{T}}(x_1),FP_{\widetilde{T}}(x_2) \right)$, respectively. In particular, we will show that all variables can be chosen such that $d_b+\frac{\lambda_b}{n_b}$ is precisely in the middle of the first interval and $d_b+\frac{\lambda_b}{n_b-c_b}$ is precisely in the middle of the latter.
\item We then show that, exploiting the first step, we can choose a value of $s_b>0$ small enough to guarantee that \emph{all} $FP_T$ values of $x \in X_b$ are strictly contained in the first interval, and \emph{all} $FP_{\widetilde{T}}$ values of $x \in X_b$ are strictly contained in the latter.
\end{itemize}

Now we proceed as follows:

\begin{itemize}
\item We set $\lambda_a:=1$. 
\item We introduce an auxiliary variable $f$ and set \\$f:=\underbrace{FP_{\widetilde{T_a}}(x_1)-FP_{T_a}(x_1)}_{\geq 0\mbox{ \tiny by Lemma \ref{lem_deletionincrease}}}+\underbrace{FP_{\widetilde{T_a}}(x_2)-FP_{T_a}(x_2)}_{>0 \mbox{ \tiny by Lemma \ref{lem_onlyoneunchanged}}}>0$. 
\item Set $\lambda_b:=\left(\frac{1}{2}f +\frac{c_a}{n_a(n_a-c_a)}\right) \frac{n_b(n_b-c_b)}{c_b}$. Note that as $f,n_a,n_b,c_a,c_b>0$ and $n_a>c_a$ and $n_b>c_b$, we have $\lambda_b>0$. 
\item We introduce another auxiliary variable $t$ and set \\$t:= \frac{1}{2}f\frac{n_b-c_b}{c_b}+\frac{c_a}{n_a(n_a-c_a)}\cdot\frac{n_b-c_b}{c_b}-\frac{1}{n_a} - \frac{1}{2}FP_{T_a}(x_1)-\frac{1}{2}FP_{T_a}(x_2)$. 
\item Set $d_a:=\begin{cases} 1 & \mbox{ if } t <0 \\ t+1& \mbox{ else. } \end{cases}$ \\Note that by this choice of $d_a$, we can guarantee both that $d_a>0$ and $d_a>t$. 
\item Set $d_b:=\frac{1}{2}FP_{T_a}(x_1)+\frac{1}{2}FP_{T_a}(x_2)+d_a+\frac{1}{n_a}-\frac{\lambda_b}{n_b}$. We need to show that $d_b>0$. In order to do so, we first substitute the chosen value for  $\lambda_b$ into the definition of $d_b$ and get: \begin{equation}\label{eq_dblambdab}d_b=\frac{1}{2}FP_{T_a}(x_1)+\frac{1}{2}FP_{T_a}(x_2)+d_a+\frac{1}{n_a}- \frac{1}{n_b} \underbrace{\left( \left(\frac{1}{2}f +\frac{c_a}{n_a(n_a-c_a)}\right) \frac{n_b(n_b-c_b)}{c_b} \right)}_{=\lambda_b}. \end{equation} This term is larger than 0 if and only if $ d_a >   \left(\frac{1}{2}f +\frac{c_a}{n_a(n_a-c_a)}\right) \frac{n_b-c_b}{c_b} -\frac{1}{n_a}-\frac{1}{2}FP_{T_a}(x_1)-\frac{1}{2}FP_{T_a}(x_2).$ However, as the right-hand side of the latter inequality equals $t$ and as $d_a>t$ as shown above, this proves that indeed $d_b>0$. 
\end{itemize}

So with these values of $\lambda_a$, $\lambda_b$, $d_a$ and $d_b$, we have that $$ d_b+\frac{\lambda_b}{n_b}\overset{\eqref{eq_dblambdab}}{=}\frac{1}{2}FP_{T_a}(x_1)+\frac{1}{2}FP_{T_a}(x_2)+d_a+\frac{1}{n_a}$$ $$ = FP_{T_a}(x_2)+d_a+\frac{1}{n_a} + \frac{1}{2}\left( \left(FP_{T_a}(x_1)+d_a+\frac{1}{n_a}\right) -\left(FP_{T_a}(x_2)+d_a+\frac{1}{n_a} \right)\right)$$
$$\overset{\eqref{eq_Ta}}{=} FP_T(x_2)+\frac{1}{2}\left(FP_T(x_1)-FP_T(x_2)\right),$$
where the last equation uses $\lambda_a=1$.

Thus,  $d_b+\frac{\lambda_b}{n_b}$ is precisely the middle of the interval $\left( FP_{T}(x_2), FP_{T}(x_1)\right)$.

Analogously, again using $\lambda_a=1$, we have
\begin{align*}
d_b+\frac{\lambda_b}{n_b-c_b}&\overset{\eqref{eq_Ta}}{=}\frac{1}{2}FP_{T_a}(x_1)+\frac{1}{2}FP_{T_a}(x_2)+d_a+\frac{1}{n_a}-\frac{\lambda_b}{n_b} +\frac{\lambda_b}{n_b-c_b} \\
&= \frac{1}{2}FP_{T_a}(x_1)+\frac{1}{2}FP_{T_a}(x_2)+d_a+\frac{1}{n_a}+\lambda_b\frac{c_b}{n_b(n_b-c_b)} \\
&= \frac{1}{2}FP_{T_a}(x_1)+\frac{1}{2}FP_{T_a}(x_2)+d_a+\frac{1}{n_a} \\
&\qquad +\underbrace{\left( \left(\frac{1}{2}f +\frac{c_a}{n_a(n_a-c_a)}\right) \frac{n_b(n_b-c_b)}{c_b} \right)}_{=\lambda_b}\frac{c_b}{n_b(n_b-c_b)}\\
&=  \frac{1}{2}FP_{T_a}(x_1)+\frac{1}{2}FP_{T_a}(x_2)+d_a+\frac{1}{n_a}+\frac{c_a}{n_a(n_a-c_a)}+\frac{1}{2}f \\
&= \frac{1}{2}FP_{T_a}(x_1)+\frac{1}{2}FP_{T_a}(x_2)+d_a+\frac{1}{n_a}+\frac{c_a}{n_a(n_a-c_a)} \\
&\qquad +\frac{1}{2}    \underbrace{\left(   FP_{\widetilde{T_a}}(x_1)-FP_{T_a}(x_1)+FP_{\widetilde{T_a}}(x_2)-FP_{T_a}(x_2)  \right)}_{=f} \\
&= \frac{1}{2}FP_{\widetilde{T_a}}(x_1)+ \frac{1}{2}FP_{\widetilde{T_a}}(x_2)+d_a+\frac{1}{n_a}+\underbrace{\left(\frac{1}{n_a-c_a}-\frac{1}{n_a}\right)}_{=\frac{c_a}{n_a(n_a-c_a)}} \\
&= \frac{1}{2}FP_{\widetilde{T_a}}(x_1)+ \frac{1}{2}FP_{\widetilde{T_a}}(x_2)+d_a +\frac{1}{n_a-c_a} \\
&= FP_{\widetilde{T_a}}(x_1)+d_a+\frac{1}{n_a-c_a} \\
&\qquad+\frac{1}{2} \left( \left(FP_{\widetilde{T_a}}(x_2)+d_a+\frac{1}{n_a-c_a}  \right) -\left( FP_{\widetilde{T_a}}(x_1)+d_a+\frac{1}{n_a-c_a} \right)\right)\\
&\overset{\eqref{eq_limTatilde}}{=} FP_{\widetilde{T}}(x_1) + \frac{1}{2}\left( FP_{\widetilde{T}}(x_2)-FP_{\widetilde{T}}(x_1)\right).
\end{align*}

Thus,  $d_b+\frac{\lambda_b}{n_b-c_b}$ is precisely the middle of the interval $\left( FP_{\widetilde{T}}(x_1), FP_{\widetilde{T}}(x_2)\right)$.

\par \vspace{0.5cm}
In total, this implies by Equations \eqref{eq_limit} and \eqref{eq_limit2} that for a small enough value of $s_b>0$, \emph{all} taxa $x \in X_b$ will have $FP_T(x)$ in the interval $(FP_T(x_2), FP_T(x_1))$ as well as $FP_{\widetilde{T}}(x)$ in the interval $(FP_{\widetilde{T}}(x_1), FP_{\widetilde{T}}(x_2))$. 

\par \vspace{0.5cm}
However, in order for our proof to be entirely constructive, we now show how to find such values of $s_b$. Therefore, first note that by Equations \eqref{eq_limit} and \eqref{eq_limit2} we know that $FP_T(x)$ and $FP_{\widetilde{T}}(x)$ for all $x\in X_b$ are strictly larger than the middle of their respective intervals $(FP_T(x_2),FP_T(x_1))$ and $(FP_{\widetilde{T}}(x_1),FP_{\widetilde{T}}(x_2))$. So we only have to make sure that these values do not get larger than the upper bound of these intervals. 

Now, let $x_3, x_4 \in X_b$ be as follows: $x_3=\argmax\limits_{x \in X_b}FP_{T_b}(x)$ and $x_4=\argmax\limits_{x \in \widetilde{X_b}\subseteq X_b}FP_{\widetilde{T_b}}(x)$. This means that $x_3$ has the highest possible $FP_{T_b}$ value  and $x_4$ has the highest $FP_{\widetilde{T_b}}$ value. So if we ensure that these two maxima are still smaller than the upper bounds of the intervals $(FP_T(x_2),FP_T(x_1))$ and $(FP_{\widetilde{T}}(x_1),FP_{\widetilde{T}}(x_2))$, respectively, this will complete the proof.

So we need to choose $s_b>0$ such that $FP_{T}(x_3)<FP_T(x_1)$ and $FP_{\widetilde{T}}(x_4)< FP_{\widetilde{T}}(x_2)$. 

The first inequality holds if and only if $FP_{T_b}(x_3)\cdot s_b + d_b + \frac{\lambda_b}{n_b} < FP_{T_a}(x_1) + d_a + \frac{ \lambda_a}{n_a}$ (recall, however, that we set $\lambda_a=1$). Using our above choice of $d_b$, it is easy to see that this holds if and only if \begin{equation}\label{eq_sb1} s_b< \frac{FP_{T_a}(x_1)-FP_{T_a}(x_2)}{2FP_{T_b}(x_3)}. \end{equation} Note that the right-hand side of this inequality is strictly positive as we have $FP_{T_a}(x_1)>FP_{T_a}(x_2)$ by assumption.

However, as explained above, we also need $FP_{\widetilde{T}}(x_4)< FP_{\widetilde{T}}(x_2)$. This holds precisely if $FP_{\widetilde{T_b}}(x_4)\cdot s_b + d_b + \frac{\lambda_b}{n_b-c_b} < FP_{\widetilde{T_a}}(x_2) + d_a + \frac{\lambda_a}{n_a-c_a}$. Without further simplifying this term using our choices of variables from above, we now simply note that the latter inequality holds if and only if \begin{equation}\label{eq_sb2}s_b < \frac{FP_{\widetilde{T_a}}(x_2)+d_a+\frac{\lambda_a}{n_a-c_a}-d_b-\frac{\lambda_b}{n_b-c_b}}{FP_{\widetilde{T_b}}(x_4)}. \end{equation}

We now need to show that the right-hand side of Equation \eqref{eq_sb2} is strictly positive. This is true if and only if $FP_{\widetilde{T_a}}(x_2)+d_a+\frac{\lambda_a}{n_a-c_a}>d_b+\frac{\lambda_b}{n_b-c_b}$. However, recall that $d_b+\frac{\lambda_b}{n_b-c_b}$ is the middle of the interval $\left(FP_{\widetilde{T}}(x_1),FP_{\widetilde{T}}(x_2) \right)$, so it is strictly smaller than the upper bound of this interval, which is why we have $FP_{\widetilde{T}}(x_2) > d_b+\frac{\lambda_b}{n_b-c_b}$. In particular, this implies the desired inequality $FP_{\widetilde{T_a}}(x_2)+d_a+\frac{\lambda_a}{n_a-c_a}>d_b+\frac{\lambda_b}{n_b-c_b}$ as both $d_a$ and $\frac{\lambda_a}{n_a-c_a}$ are larger than 0.

\par\vspace{0.5cm}
Now we introduce an auxiliary variable $m$ and let \\$m:= \min \left\{  \frac{FP_{T_a}(x_1)-FP_{T_a}(x_2)}{2FP_{T_b}(x_3)}, \frac{FP_{\widetilde{T_a}}(x_2)+d_a+\frac{\lambda_a}{n_a-c_a}-d_b-\frac{\lambda_b}{n_b-c_b}}{FP_{\widetilde{T_b}}(x_4)}  \right\},$ where $\lambda_a, \lambda_b, d_a$ and $d_b$ are chosen as above. Then, $m>0$, and we can take any value from the open interval $(0,m)$ as possible values for $s_b$. For instance, we can simply choose $s_b:= \frac{1}{2}m$. By the definition of $m$, this will fulfill both \eqref{eq_sb1} and \eqref{eq_sb2}, and it will still be positive, which means this value of $s_b$ is a valid scaling factor for the edge lengths in $T_b$. 

Now that we have chosen all edge lengths of $T$ by downscaling $T_b$ and extending the pendant edge lengths of $T_a$ by constant $d_a$ and those of the scaled version of $T_b$ by constant $d_b$, we have achieved that the $FP_T$ values of all $x \in X_b$ are in between $FP_T(x_2)$ and $FP_T(x_1)$ in the exact order induced by $\pi_{T_b}$. Analogously, the $FP_{\widetilde{T}}$ values of all $x \in X_b$ are in between $FP_{\widetilde{T}}(x_1)$ and $FP_{\widetilde{T}}(x_2)$ in the exact order induced by $\pi_{T_b}$. Moreover, the  $FP_T$ and $FP_{\widetilde{T}}$ values of all $x \in X_a$ behave exactly as induced by $\pi_{T_a}$. Therefore, these edge lengths are in total such that the derived ranking is strict and reversible. 

Next, we argue why this strict and reversible ranking is such that the taxon with the highest $FP_T$ value is still contained in $\widetilde{T}$, while the taxon with the smallest $FP_T$ value is not. Note that $\pi_{T_a}$ is strict and reversible and fulfills all requirements of the theorem by the inductive hypothesis. So the leaf deletion used to reverse the ranking of $FP_{T_a}$ does not affect leaf $x$ with $x=\argmax\limits_{y \in X_a} FP_{T_a}(y)$, while the deletion of leaf $y$ with $y=\argmin\limits_{z \in X_a} FP_{T_a}(z)$ is ensured. So leaf $x$ is also still present in $\widetilde{T}$ while $y$ is not. Moreover, $x$ must have maximal $FP_T$ value and $y$ must have minimal $FP_T$ value, because the $FP_T$ value of all leaves of $T_b$ are now by construction between $FP_T(x_1)$ and $FP_T(x_2)$ and can therefore neither be maximal nor minimal for $T$. So the maximal leaf of $T_a$ must be the maximal leaf of $T$ and the minimal leaf of $T_a$ must also be the minimal leaf of $T$, as $FP_T(u)>FP_T(v) \Leftrightarrow FP_{T_a}(u)>FP_{T_a}(v)$ for all $u,v \in X_a$ by Equation \eqref{eq_Ta}. This completes the proof.
\end{enumerate}
\end{proof}


\setcounter{thm}{2}
\begin{thm}
Let $T$ be a rooted binary phylogenetic $X$-tree with $|X|=n \geq 2$. 
Then, there exist strictly positive edge lengths $\lambda_1, \ldots, \lambda_{2n-2}$ for $T$ such that 
\begin{enumerate}[(i)]
\item the ranking $\pi_T$ induced by the FP index for the leaves of $T$ is strict, and 
\item deleting leaf $y \coloneqq \argmin\limits_{x \in X} FP_T(x)$ from $T$ results in a strict ranking $\pi_{\widetilde{T}}$ for tree $\widetilde{T} \coloneqq T \setminus \{y\}$ on leaf set $\widetilde{X} = X \setminus \{y\}$, for which $w \coloneqq \argmax\limits_{x \in \widetilde{X}} FP_{\widetilde{T}}(x) = \argmin\limits_{x \in \widetilde{X}} FP_T(x)$.
\end{enumerate}
In words, there exist strictly positive edge lengths for $T$ such that if the species with the lowest $FP_T$ value goes extinct, the species with the second lowest $FP_T$ value has the highest $FP_{\widetilde{T}}$ value. 
\end{thm}

In order to prove Theorem \ref{thm_secondleast_becomes_first}, we require two more lemmas.  The first one shows that the FP indices of leaves in a maximal pendant subtree that is not affected by a leaf deletion will not change.

\begin{lem}\label{lem_nonaffectedsubtree} Let $T=(T_a,T_b)$ be a binary phylogenetic $X$-tree with $|X|=n\geq2$. Let $\widetilde{T}$ be a tree that results from $T$ by deleting some, but \emph{not all}, leaves from $T_a$ (or $T_b$) only. Then, we have $FP_{\widetilde{T}}(x)=FP_{T}(x)$ for all $x \in T_b$ (or $T_a$, respectively).
\end{lem}

\begin{proof} Let $T=(T_a,T_b)$ and $\widetilde{T}$ be as described in the theorem. Assume without loss of generality that the deleted leaf set $X'$ is a strict subset of the taxon set $X_a$ of $T_a$. Let $\widetilde{T_a}$ and $\widetilde{T_b}$ be the trees that result from $T_a$ and $T_b$, respectively when the leaves of $X'\subset X_a$ are deleted, and denote by $\lambda_b$ the length of the edge leading to $T_b$ in $T$ and by $n_b$ the number of leaves in $T_b$. Then, $FP_T(x)=FP_{T_b}(x) + \frac{\lambda_b}{n_b}$ and $FP_{\widetilde{T}}(x)=FP_{\widetilde{T_b}}(x) + \frac{\lambda_b}{n_b}$ for all $x \in X_b$. However, as no leaf from $T_b$ got deleted and not all leaves from $T_a$ have been deleted so that the edge leading to $T_b$ does not get suppressed, we have $\widetilde{T_b}=T_b$,  which shows that $FP_T(x)=FP_{\widetilde{T}}(x)$ for all $x \in X_b$.  This completes the proof. 
\end{proof}

The next lemma shows that any vector of strictly positive real numbers can be realized as a vector of FP indices on a phylogenetic tree $T$ with strictly positive edge lengths.

\begin{lem}\label{lem_strictphi} Let $T$ be a rooted binary phylogenetic tree on leaf set $X=\{x_1,\ldots,x_n\}$, and let $\theta=(\theta_1,\ldots,\theta_n) \in \mathbb{R}^n_+$ be a vector of strictly positive values. Then, there exist strictly positive edge lengths for $T$ such that the induced FP indices equal $\theta$, i.e. $FP_T({x_i}) =\theta_i$ for all $i=1,\ldots,n$. 
\end{lem}

The statement of Lemma \ref{lem_strictphi} is established in \citet{Wicke2020} (proof of Theorem 2 therein). Note that Theorem 2 in \citet{Wicke2020} itself is phrased in a slightly different context, but in the proof it is shown that given a rooted phylogenetic $X$-tree with $X=\{x_1, \ldots, x_n\}$ and a vector  $\theta=(\theta_1,\ldots,\theta_n) \in \mathbb{R}^n_+$ of strictly positive real numbers, there exists an edge length assignment $\lambda: E(T) \rightarrow \mathbb{R}_{+}$ that assigns strictly positive lengths to all edges of $T$ such that $\theta_i = FP_{(T, \lambda)}(x_i)$ for all $x_i \in X$.

We are now in the position the prove Theorem \ref{thm_secondleast_becomes_first}.

\begin{proof}[Proof of Theorem \ref{thm_secondleast_becomes_first}]
We prove this statement by induction on $n$. If $n=2$, there is only one rooted binary phylogenetic $X$-tree, namely the one that consists of a cherry, say $[x_1,x_2]$, and the statement trivially holds.

We now assume that the statement holds for all trees with at most $n-1$ leaves and let $T=(T_a,T_b)$ be a tree with $n \geq 3$ leaves (else consider the base case). Without loss of generality we may assume that $n_a \geq n_b$. In particular, $n_a \geq 2$ (because $n \geq 3$). Moreover, we may assume that $X_a = \{x_1, \ldots, x_{n_a}\}$ and $X_b = \{x_{n_a+1}, \ldots, x_{n_a+n_b}\}$ (else, re-label the leaves of $T$ accordingly).

As $T_a$ is a rooted binary phylogenetic tree with $n_a < n$ leaves, by the inductive hypothesis, there exist strictly positive edge lengths for $T_a$ that satisfy all requirements stated by the theorem. We fix these edge lengths and additionally assign length $\lambda_a \in \mathbb{R}_{+}$ to edge $(\rho,a)$ connecting $T_a$ with the root of $T$. Then, $FP_{T}(x) = FP_{T_a}(x)+\frac{\lambda_a}{n_a}$ for all $x \in X_a$.

Without loss of generality, we may assume that $FP_{T_a}(x_1) = \argmin\limits_{x \in X_a} FP_{T_a}(x)$ and $FP_{T_a}(x_2) = \argmin\limits_{x \in X_a \setminus \{x_1\}} FP_{T_a}(x)$ (else, re-label the leaves in $T_a$ accordingly), i.e. leaves $x_1$ and $x_2$ have the lowest and second lowest $FP_{T_a}$ value among the taxa in $X_a$, respectively. 

Let $\widetilde{T}$ be the tree obtained from $T$ by deleting leaf $x_1$ and let $\widetilde{X} \coloneqq X \setminus \{x_1\}$ denote the leaf set of $\widetilde{T}$. Note that by the inductive hypothesis, $x_2 = \argmax\limits_{x \in \widetilde{X_a}} FP_{\widetilde{T_a}}(x)$.

We now choose the edge lengths of $T_b$ as well as the length $\lambda_b$ of edge $(\rho,b)$ connecting  $T_b$ with the root of $T$ such that $FP_T(x) \in (FP_{T_a}(x_2)+\frac{\lambda_a}{n_a}, FP_{\widetilde{T_a}}(x_2) + \frac{\lambda_a}{n_a-1}) = (FP_T(x_2), FP_{\widetilde{T}}(x_2))$ for all $x \in X_b$ and such that $FP_T(x_i) \neq FP_T(x_j)$ as well as $FP_{\widetilde{T}}(x_i) \neq FP_{\widetilde{T}}(x_j)$ $\forall x_i \neq x_j\, \in \{x_1, \ldots, x_n\}$.

To see that this is possible, first note that $FP_T(x_2) \neq FP_{\widetilde{T}}(x_2)$:
\begin{itemize}
    \item If $n_a=2$, i.e. $x_1$ and $x_2$ are the only leaves in $T_a$, we have
    \begin{align*}
        FP_T(x_2) &= \lambda_{e_2} + \frac{\lambda_a}{2} \quad \text{and} \\
        FP_{\widetilde{T}} (x_2) &= \lambda_{e_2} + \lambda_a > FP_T(x_2),
    \end{align*}
    where $\lambda_{e_2}$ denotes the length of the pendant edge incident with $x_2$.
    \item If $n_a>2$, then $FP_{\widetilde{T_a}}(x_2) \geq FP_{T_a}(x_2)$ by Lemma \ref{lem_deletionincrease} and $\frac{\lambda_a}{n_a-1} > \frac{\lambda_a}{n_a}$. In particular, $FP_T(x_2) < FP_{\widetilde{T}}(x_2)$.
    \end{itemize}
Thus, in both cases the interval $(FP_T(x_2), FP_{\widetilde{T}}(x_2))$ is non-empty and contains infinitely many strictly positive real numbers. 

We now distinguish two cases:
\begin{enumerate}[(a)]
\item If $T_b$ consists of precisely one leaf, namely leaf $x_n$, we choose $\lambda_b$, which in this case equals  $FP_T(x_n) = FP_{\widetilde{T}}(x_n)$, in the interval  $ (FP_T(x_2), FP_{\widetilde{T}}(x_2))$,  such that $ \lambda_b \neq FP_T(x)$ for all $x \in X_a$ and $\lambda_b \neq FP_{\widetilde{T}}(x)$ for all $x \in X_a$.

\item If $T_b$ contains $n_b > 1$ leaves, we choose $\varepsilon$ such that $0 < \frac{\varepsilon}{n_b} < FP_T(x_2)$ and set $\lambda_b = \varepsilon$.
We then choose $n_b$ distinct strictly positive real values $\theta_1, \ldots, \theta_{n_b} \in (FP_T(x_2)-\frac{\varepsilon}{n_b}, FP_{\widetilde{T}}(x_2) - \frac{\varepsilon}{n_b})$ such that $\theta_i \neq FP_T(x)$ for all $i$ and all $x \in X_a$ and $\theta_i \neq FP_{\widetilde{T}}(x)$ for all $i$ and all $x \in X_a$. Then, by Lemma \ref{lem_strictphi} we can choose strictly positive edge lengths for $T_b$ such that the induced $FP_{T_b}$ values for the leaves in $X_b$ (i.e. for leaves $x_{n_a+1}, \ldots, x_{n_a+n_b}$) equal $\theta$, i.e. $FP_{T_b}(x_{n_a+i}) = \theta_i$ for $i=1, \ldots, n_b$. Now, by choice of $\theta_i$ and $\lambda_b$, this implies that $FP_T(x_{n_a+i}) = FP_{T_b}(x_{n_a+i}) + \frac{\varepsilon}{n_b} \in (FP_T(x_2), FP_{\widetilde{T}}(x_2))$ for all $i \in \{1, \ldots,n_b\}$.
\end{enumerate}

Thus, in both cases, we can choose the edge lengths of $T_b$ and the length of edge $(\rho, b)$ such that $FP_T(x) \in (FP_T(x_2), FP_{\widetilde{T}}(x_2))$ for all $x \in X_b$ and such that the rankings $\pi_T$ and $\pi_{\widetilde{T}}$ are strict. \\

Now, as $FP_T(x) > FP_T(x_2)$ for all $x \in X_b$ by construction, and $FP_T(x_2) > FP_T(x_1)$ (because by the inductive hypothesis for $T_a$ and the assumption that $x_1 = \argmin\limits_{x \in X_a} FP_{T_a}(x)$, we have $FP_{T_a}(x_1) < FP_{T_a}(x_2)$ and thus also $FP_T(x_1) < FP_T(x_2))$, we can conclude that $x_1 = \argmin\limits_{x \in X} FP_T(x)$. Similarly, as by assumption $x_2 = \argmin\limits_{x \in X_a \setminus \{x_1\}} FP_{T_a}(x)$, we can conclude that $x_2 = \argmin\limits_{x \in \widetilde{X}} FP_{T}(x)$. Thus, leaves $x_1$ and $x_2$ have the lowest, respectively second lowest $FP_T$ value. \\

Moreover, as $FP_T(x) = FP_{\widetilde{T}}(x) < FP_{\widetilde{T}}(x_2)$ for all $x \in X_b$ by construction (where the first equality follows from the fact that the FP indices of leaves in $X_b$ are not affected by the deletion of leaf $x_1 \in X_a$; cf. Lemma \ref{lem_nonaffectedsubtree}) and $FP_{\widetilde{T}}(x) < FP_{\widetilde{T}}(x_2)$ for all $x \in \widetilde{X_a}$ by the inductive hypothesis (as $x_2 = \argmax\limits_{x \in \widetilde{X_a}} FP_{\widetilde{T_a}}(x)$ and thus in particular, $FP_{\widetilde{T}}(x_2) > FP_{\widetilde{T}}(x)$ for all $x \in X_a \setminus \{x_2\}$ because the ranking is strict), we can conclude that $x_2 = \argmax\limits_{x \in \widetilde{X}} FP_{\widetilde{T}}(x)$, i.e. leaf $x_2$ has the largest $FP_{\widetilde{T}}$ value.\\

Thus, we have shown that there are edge lengths for $T$ such that if the leaf with the lowest $FP_T$ value (in our case leaf $x_1$) is deleted, the leaf with the second lowest $FP_T$ value (in our case leaf $x_2$) has the highest $FP_{\widetilde{T}}$ value. Moreover, the corresponding rankings are strict. This completes the proof.
\end{proof}

\setcounter{thm}{3}
\begin{thm} 
Let $T=(T_a,T_b)$ be a rooted binary phylogenetic $X$-tree with $|X|=n\geq 3$ such that $T_a$ and $T_b$ have $n_a$ and $n_b$ leaves, respectively, where $n_a\geq n_b$. Let $f:E(T)\longrightarrow \mathbb{R}_+$ be some function that assigns all edges of $T$ positive edge lengths. Let $x^*:=\argmax\limits_{x \in X} FP_T(x) $ be the leaf of $T$ with the highest $FP_T$ value concerning the edge length assignment of $f$. Then, we have: \begin{enumerate}
    \item If a leaf $x'$ other than $x^*$ is deleted (e.g. the leaf with minimal $FP_T$ value) to derive a tree $\widetilde{T}$, we denote the number of leaves that have a higher $FP_{\widetilde{T}}$ value than $x^*$ in $\widetilde{T}$ by $\mathcal{N}$ and have that $\mathcal{N}\leq n_a-1$ (with $ \lfloor \frac{n-1}{2}\rfloor \leq n_a-1$).
\item There exists an edge length assignment $\widehat{f}$ such that this bound is achieved, i.e. $\mathcal{N}= n_a-1$, and the resulting ranking $\pi_T$ is strict.
\end{enumerate}
\end{thm}

In order to prove Theorem \ref{thm_boundoneleaf},  we require one more lemma which complements Lemma \ref{lem_nonaffectedsubtree} as it shows that when only leaves from one maximal pendant subtree are deleted, while for \emph{none} of the leaves of the other maximal pendant subtree the FP indices change (Lemma \ref{lem_nonaffectedsubtree}), they increase for \emph{all} leaves of the original subtree.

\begin{lem}\label{lem_affectedsubtree} Let $T=(T_a,T_b)$ be a binary phylogenetic $X$-tree with $|X|=n\geq2$. Let $\widetilde{T}$ be a tree that results from $T$ by deleting some (but not all) leaves from $T_a$ (or $T_b$) only. Then, we have $FP_{\widetilde{T}}(x)>FP_{T}(x)$ for all $x \in T_a$ (or $T_b$, respectively).
\end{lem}

\begin{proof} Let $T=(T_a,T_b)$ and $\widetilde{T}$ be as described in the theorem. Assume without loss of generality that the deleted leaf set $X'$ is a proper subset of the taxa of $T_a$, and denote by $\kappa_a:=|X'|>0$ the number of deleted leaves; i.e. we have $X'\subset X_a$ and thus $\kappa_a<n_a$, where $n_a$ denotes the number of leaves in $T_a$. Let $\widetilde{T_a}$ be the tree that results from $T_a$ when the leaves of $X'$ are deleted, and denote by $\lambda_a$ the length of the edge leading to $T_a$ in $T$. Then, $FP_T(x)=FP_{T_a}(x) + \frac{\lambda_a}{n_a}$ and $FP_{\widetilde{T}}(x)=FP_{\widetilde{T_a}}(x) + \frac{\lambda_a}{n_a-\kappa_a}$ for all $x \in X_a$. By Lemma \ref{lem_deletionincrease} we have $FP_{T_a}(x)\leq FP_{\widetilde{T_a}}(x)$. So in total, this implies $$FP_T(x)=FP_{T_a}(x) + \frac{\lambda_a}{n_a} < FP_{\widetilde{T_a}}(x) + \frac{\lambda_a}{n_a-\kappa_a}=FP_{\widetilde{T}}(x),$$ where the strict inequality stems from the fact that $\kappa_a>0$.  This completes the proof. 
\end{proof}

We are now in the position to prove Theorem \ref{thm_boundoneleaf}.

\begin{proof} [Proof of Theorem \ref{thm_boundoneleaf}] Let $T$, $f$, $x^*$ and $x'$ be as stated in the theorem and let $\lambda_a$, $\lambda_b$ denote the lengths of the edges leading to $T_a$ and $T_b$, respectively. Note that if $x'$ is in $T_a$, we have $FP_T(x)=FP_{\widetilde{T}}(x)$ for all $x$ in $T_b$ by Lemma \ref{lem_nonaffectedsubtree}.  However, for all leaves $x$ in $T_a$ (other than the deleted $x'$) we have $FP_{T}(x)<FP_{\widetilde{T}}(x)$ by Lemma \ref{lem_affectedsubtree}. So depending on whether $x'$ is in $T_a$ or $T_b$, we have $n_a-1$ or $n_b-1$ leaves whose FP indices strictly increase when $x'$ is deleted and all others remain unchanged. As $x^*$ is the taxon with maximal $FP_T$ value, the only taxa which can have a larger $FP_{\widetilde{T}}$ value than $x^*$ are the ones whose FP indices increases when $x'$ is deleted, so these are at most $n_a-1$ leaves (as $n_a-1\geq n_b-1$). Moreover, as $n_a \geq n_b$ and $n_a+n_b=n$, we have $n_a\geq \lfloor \frac{n+1}{2}\rfloor$ and thus $n_a-1 \geq \lfloor\frac{n+1}{2} \rfloor-1 =\lfloor\frac{n-1}{2} \rfloor$. This completes the first part of the proof.

For the second part of the proof we use Lemma \ref{lem_strictphi} to assign positive edge lengths to $T_a$ and $T_b$ that induce strict rankings $\pi_{T_a}$ and $\pi_{T_b}$.  We denote by $X_a$ and $X_b$ the leaf sets of $T_a$ and $T_b$, respectively. Then, we fix leaves $x^*_a:=\argmax\limits_{x \in X_a} FP_{T_a}(x)$ and $x^*_b:=\argmax\limits_{x \in X_b} FP_{T_b}(x)$, i.e. $x^*_a$ is the leaf with highest $FP_{T_a}$ value and $x^*_b$ is the leaf with highest $FP_{T_b}$ value. Similarly, we fix leaves $x'_a:=\argmin\limits_{x \in X_a} FP_{T_a}(x)$ and $x'_b:=\argmin\limits_{x \in X_b} FP_{T_b}(x)$, i.e. $x'_a$ is the leaf with smallest $FP_{T_a}$ value and $x'_b$ is the leaf with smallest $FP_{T_b}$ value. We will now proceed as follows: Using Corollary \ref{cor_lem3lem4}, we scale $T_a$ and $T_b$ by positive scaling factors $s_a$ and $s_b$, respectively, such that the rankings induced by $\pi_{T_a}$ and $\pi_{T_b}$ are not changed. If $\lambda_a$ and $\lambda_b$ denote the edge lengths of the edges leading from the root of $T$ to $T_a$ and $T_b$, respectively, we derive the following equalities: $$FP_T(x)=FP_{T_a}(x)\cdot s_a + \frac{\lambda_a}{n_a} \stackrel{s_a\rightarrow 0}{\longrightarrow} \frac{\lambda_a}{n_a} \mbox{  for all  } x \in X_a, \mbox{ and}$$
 $$FP_T(x)=FP_{T_b}(x)\cdot s_b + \frac{\lambda_b}{n_b}\stackrel{s_b\rightarrow 0}{\longrightarrow} \frac{\lambda_b}{n_b} \mbox{  for all  } x \in X_b.$$
 
 Moreover, if we delete $x'_a$ from $T$ and thus from $T_a$ to derive tree $\widetilde{T}$, the FP indices of $T_b$ remain completely unchanged by Lemma \ref{lem_nonaffectedsubtree}, so we have $FP_T(x)=FP_{\widetilde{T}}(x)$ for all $x \in X_b$. However, the FP indices of \emph{all} taxa of $T_a$ strictly increase by Lemma \ref{lem_affectedsubtree}, as the proportion of $\lambda_a$ assigned to all leaves of $T_a$ increases from $\frac{\lambda_a}{n_a}$ to $\frac{\lambda_a}{n_a-1}$ (and possibly the proportion of other edges within $T_a$ that is assigned to taxa from $T_a$ also gets higher, so we can only state the following lower bound for the $FP_{\widetilde{T}}$ values), so we have: 
 
 \begin{equation} \label{eq_lowerboundFPtildeT}FP_{\widetilde{T}}(x) \geq FP_{T}(x) + \frac{\lambda_a}{n_a(n_a-1)} \hspace{0.3cm}\forall x\in X_a\setminus \{x'_a\}.\end{equation} The latter summand is due to the fact that $\frac{\lambda_a}{n_a-1}=\frac{\lambda_a}{n_a}+\frac{\lambda_a}{n_a(n_a-1)}$ for all $x \in X_a\setminus \{x'_a\}$.

We now want to choose positive values for $\lambda_a$, $\lambda_b$, $s_a$ and $s_b$ such that $FP_T(x_b')>FP_T(x_a^*)$ and $FP_{\widetilde{T}}(x_b^*)<FP_{\widetilde{T}}(x_a)$ for all $x_a\in X_a\setminus \{x'_a\}$. Note that these inequalities imply that \emph{all} leaves of $T_b$ have higher $FP_T$ values than the leaves of $T_a$ and that \emph{all} leaves of $T_b$ have lower $FP_{\widetilde{T}}$ values than the leaves of $T_a$ (other than $x_a'$). 

\begin{itemize}
\item We start by setting $\lambda_a:=1$. 
\item Next, we set $\lambda_b:=\frac{(2n_a-1)n_b}{2n_a(n_a-1)}$. Note that $\lambda_b>0$ as $n\geq 3$ and thus $n_a>1$. 
\item We set $s_b:=\frac{n_b-\lambda_b(n_a-1)}{2FP_{T_b}(x_b^*)\cdot n_b (n_a-1)}$. We now show that this ensures that $s_b>0$. First, as before we have $n_a>1$ (because $n\geq 3$), so the denominator is positive. Moreover, we have $n_b-\lambda_b(n_a-1)=n_b-\frac{(2n_a-1)n_b}{2n_a(n_a-1)}(n_a-1)=n_b \cdot \left( 1-\frac{2n_a-1}{2n_a}\right)= \frac{n_b}{2n_a}>0,$ so the enumerator is also positive. Thus, $s_b>0$. 

\item Last, we set $s_a:=\frac{1}{2FP_{T_a}(x_a^*)}\left(\frac{\lambda_b}{n_b}-\frac{1}{n_a}\right)$. We now show that $s_a>0$. Note that this is true precisely if $\lambda_b>\frac{n_b}{n_a}$. Using our choice of $\lambda_b$, this holds if and only if $$ \frac{(2n_a-1)n_b}{2n_a(n_a-1)}>\frac{n_b}{n_a},$$ which in turn holds precisely if 

$$ \frac{2n_a-1}{2(n_a-1)}>1.$$ As $2n_a-1>2n_a-2$, the assertion holds and we therefore indeed have $s_a>0$. 
\end{itemize}

We now first prove that for our choices of $\lambda_a$, $\lambda_b$, $s_a$ and $s_b$, we have $FP_T(x_b')>FP_T(x_a^*)$. Therefore, recall that $$FP_T(x_b')=FP_{T_b}(x_b')\cdot s_b+\frac{\lambda_b}{n_b} > \frac{\lambda_b}{n_b}=\frac{(2n_a-1)}{2n_a(n_a-1)},$$ where the latter inequality is true as all edge lengths of $T_b$ are positive, which implies that all $FP_{T_b}$ values are also positive, and where the last equality uses our choice of $\lambda_b$. On the other hand, we have $$FP_T(x_a^*)=FP_{T_a}(x_a^*)\cdot s_a + \frac{\lambda_a}{n_a}= FP_{T_a}(x_a^*)\cdot\frac{1}{2FP_{T_a}(x_a^*)}\left(\frac{\lambda_b}{n_b}-\frac{1}{n_a}\right) + \frac{1}{n_a}=\frac{1}{2} \left(\frac{\lambda_b}{n_b}+ \frac{1}{n_a} \right)<\frac{\lambda_b}{n_b}.$$
Here, the second equality uses our choices of $s_a$ and $\lambda_a$, and the final inequality is due to the fact that $\lambda_b>\frac{n_b}{n_a}$ as shown above when we chose $s_a$, and thus $\frac{1}{n_a}<\frac{\lambda_b}{n_b}$. This now leads to: $$FP_T(x_a^*)< \frac{\lambda_b}{n_b}< FP_{T_b}(x_b')+\frac{\lambda_b}{n_b}=FP_{T}(x_b').$$

So the $FP_T$ value of $x_a^*$ is strictly smaller than that of $x_b'$, but $x_a^*$ has the maximum such value of all $x_a \in X_a$, and $x_b'$ has the minimum value of all $x_b\in X_b$, so in total we have that $FP_T(x_a)<FP_T(x_b)<FP_T(x_b^*)$ for all $x_a\in X_a$, $x_b \in X_b\setminus\{x_b^*\}$. In particular, the induced ranking $\pi_{T}$ is strict (as the partial rankings $\pi_{T_a}$ and $\pi_{T_b}$ are strict and as the separation between $T_a$ and $T_b$ induced by $\pi_T$ is now also strict). Moreover, we now know that $x_a'=\argmin\limits_{x \in X} FP_T(x)$ and $x_b^*=\argmax\limits_{x \in X} FP_T(x)$, i.e. the minimum and maximum can be taken over $X$, not only over $X_a$ or $X_b$, respectively.

Next, we consider  $FP_{\widetilde{T}}(x_b^*)$. Using our knowledge that the deletion of $x_a'$ does not affect any taxa in $T_b$, particularly not $x_b^*$, and using our choices of $\lambda_b$ and $s_b$, we derive: 
$$FP_{\widetilde{T}}(x_b^*)=FP_{T}(x_b^*)=FP_{T_b}(x_b^*)\cdot s_b+\frac{\lambda_b}{n_b}=FP_{T_b}(x_b^*)\cdot\frac{n_b-\lambda_b(n_a-1)}{2FP_{T_b}(x_b^*)\cdot n_b (n_a-1)}+\frac{\lambda_b}{n_b}$$

\begin{equation} \label{eq_wantedoneleaf}= \frac{n_b-\frac{(2n_a-1)n_b}{2n_a(n_a-1)}\cdot(n_a-1)}{2n_b (n_a-1)}+\frac{\frac{(2n_a-1)n_b}{2n_a(n_a-1)}}{n_b}= \frac{4n_a-1}{4n_a(n_a-1)}. \end{equation}

We now argue that the latter term is smaller than $FP_T(x_a') + \frac{1}{n_a(n_a-1)}$. In order to see this, note that $-1 < \frac{FP_{T_a}(x_a')}{FP_{T_a}(x_a^*)}$, because the right-hand side is positive. Therefore, we have $4n_a-5 < \frac{FP_{T_a}(x_a')}{FP_{T_a}(x_a^*)} + 4 (n_a-1)$ and thus, as $4n_a(n_a-1)>0$, we, we get: 

$$\frac{4n_a-5}{4 n_a(n_a-1)} < \frac{FP_{T_a}(x_a')+FP_{T_a}(x_a^*)\cdot4(n_a-1)}{FP_{T_a}(x_a^*)\cdot 4n_a(n_a-1)}.$$ We rearrange the right-hand side of this inequality to get 

$$\frac{4n_a-5}{4 n_a(n_a-1)} < \frac{FP_{T_a}(x_a')}{FP_{T_a}(x_a^*)}\cdot \frac{1}{2} \left( \frac{2n_a-1}{2n_a(n_a-1)}-\frac{1}{n_a}\right)+\frac{1}{n_a}. $$ 

Using our choices of $\lambda_b$ and then of $s_a$, this leads to: 

$$\frac{4n_a-5}{4 n_a(n_a-1)} < \frac{FP_{T_a}(x_a')}{FP_{T_a}(x_a^*)}\cdot \frac{1}{2} \left( \frac{\lambda_b}{n_b}-\frac{1}{n_a}\right)+\frac{1}{n_a}=FP_{T_a}(x_a') \cdot s_a + \frac{1}{n_a}=FP_T(x_a'). $$ 

So we have $\frac{4n_a-5}{4 n_a(n_a-1)} < FP_T(x_a')$, and thus $$\frac{4n_a-5}{4 n_a(n_a-1)} + \frac{1}{n_a(n_a-1)}< FP_T(x_a') + \frac{1}{n_a(n_a-1)}.$$

Together with \eqref{eq_wantedoneleaf}, this shows that \begin{equation}\label{eq_wantedoneleaf2}FP_{\widetilde{T}}(x_b^*) < FP_T(x_a') + \frac{1}{n_a(n_a-1)}.\end{equation} 

As $FP_{T}(x_a') < FP_T(x_a)$ for all $x_a \in X_a\setminus \{x_a'\}$ and as $FP_{\widetilde{T}}(x_a)\geq FP_T(x_a)+\frac{1}{n_a(n_a-1)}$ for all $x_a \in X_a \setminus \{x_a'\}$ by Equation \eqref{eq_lowerboundFPtildeT}  (using that $\lambda_a=1$), we have $$FP_{T}(x_a')+\frac{1}{n_a(n_a-1)}<FP_T(x_a)+\frac{1}{n_a(n_a-1)} \leq FP_{\widetilde{T}}(x_a)$$ for all $x_a \in X_a\setminus \{x_a'\}.$ So as we have $FP_{\widetilde{T}}(x_b^*)<FP_{T}(x_a')+\frac{1}{n_a(n_a-1)}$ by \eqref{eq_wantedoneleaf2}, we can finally conclude $FP_{\widetilde{T}}(x_b^*)<FP_{\widetilde{T}}(x_a)$ for all $x_a \in X_a\setminus \{x_a'\}$.

In summary, we now have seen that the given positive edge lengths induce a strict ranking such that:
\begin{itemize} \item the unique maximal $FP_T$ value is in $T_b$, namely assumed by $x_b^*$, \item the unique minimal $FP_T$ value is in $T_a$, namely assumed by $x_a'$, \item the deletion of $x_a'$ leads to a tree $\widetilde{T}$ for which \emph{all} taxa of $X_a\setminus  \{x_a'\}$ have a higher $FP_{\widetilde{T}}$ value than $x_b^*$. Thus, there are $|X_a\setminus  \{x_a'\}|=n_a-1$ taxa which are ranked higher in $\widetilde{T}$ than the original maximum $x_b^*$. 
\end{itemize}

The latter assertion shows that for the number $\mathcal{N}$ of leaves in $T$ that have a higher $FP_{\widetilde{T}}$ value than $x^*:=x_b^*$ we have $\mathcal{N}=n_a-1$. This completes the second part of the proof.  
\end{proof}


\setcounter{prop}{0}
\begin{prop} 
Let $T$ be a rooted binary ultrametric caterpillar tree on $X$ with $|X|=n \geq 3$, and let $X' \subset X$ be a subset of the leaves. Let $\widetilde{T}$ be the induced subtree of $T$ restricted to the leaves in $\widetilde{X} \coloneqq X \setminus X'$. Then, $FP_T(x_i) \geq FP_T(x_j)$ implies $FP_{\widetilde{T}}(x_i) \geq FP_{\widetilde{T}}(x_j)$ for all $x_i,x_j \in \widetilde{X}$.
\end{prop}

\begin{proof}
Let $T$ be an ultrametric caterpillar tree on $X=\{x_1, \ldots, x_n\}$ with $n \geq 3$. Without loss of generality we may assume that the leaf labels and edge lengths of $T$ are as indicated in Figure \ref{Fig_UltrametricCaterpillar}. 
In particular, the length of the pendant edge incident to $x_1$ is $\lambda_1$, and for each $x_j$ with $j=2, \ldots, n$, the length of the pendant edge incident to $x_j$ is equal to $\sum\limits_{k=1}^{j-1} \lambda_k$ (due to the fact that $T$ is ultrametric). Moreover, the number of leaves below an interior edge of length $\lambda_j$ is equal to $j$ for $j=2, \ldots, n-1$. Thus, 
\begin{align*}
    FP_T(x_1) &= \lambda_1 + \sum\limits_{k=2}^{n-1} \frac{\lambda_k}{k},
\end{align*}
and for $j=2, \ldots, n$,
\begin{align*}
    FP_T(x_j) &= \sum\limits_{k=1}^{j-1} \lambda_k + \sum\limits_{k=j}^{n-1} \frac{\lambda_k}{k}.
\end{align*}

This directly implies
\begin{align*}
    FP_T(x_1) \leq FP_T(x_2) \leq FP_T(x_3) \leq \ldots \leq FP_T(x_{n-1}) \leq FP_T(x_n).
\end{align*}
(In fact, we have $FP_T(x_1)=FP_T(x_2)$, whereas all other inequalities are strict.)
Now, suppose that a single element $x_i \in X'$ is deleted from $T$ to obtain $\widetilde{T}$. 
\begin{itemize}
    \item If $x_i=x_n$, $\widetilde{T}$ is a caterpillar tree on $\{x_1, \ldots, x_{n-1}\}$ with the parent of $x_{n-1}$ being the root.
Analogously to the calculations above it now follows that
\begin{align*}
    FP_{\widetilde{T}}(x_1) \leq FP_{\widetilde{T}}(x_2) \leq \ldots \leq FP_{\widetilde{T}}(x_{n-1}). 
\end{align*}
In particular, the ranking order of the elements in $\widetilde{X}$ on $\widetilde{T}$ is identical to the respective ranking order on $T$. 
    \item If $x_i \neq x_n$, $\widetilde{T}$ is a caterpillar tree on $\widetilde{X}=X \setminus \{x_i\}$ obtained from $T$ by deleting $x_i$ and its incident edge and suppressing the resulting degree-2 vertex. Without loss of generality we can assume that $x_i \neq x_1$ (otherwise switch the labels of $x_1$ and $x_2$) and we have
    \begin{align}
        FP_{\widetilde{T}}(x_1) &= \lambda_1 + \sum\limits_{k=2}^{i-2} \frac{\lambda_k}{k} + \frac{\lambda_{i-1} + \lambda_i}{i-1} + \sum\limits_{k=i+1}^{n-1} \frac{\lambda_k}{k-1}.
        \label{Eq1}
     \end{align}
     For $2 \leq j \leq i-1$, we have
     \begin{align}
         FP_{\widetilde{T}}(x_j) = \sum\limits_{k=1}^{j-1} \lambda_k
          + \sum\limits_{k=j}^{i-2} \frac{\lambda_k}{k} + \frac{\lambda_{i-1} + \lambda_i}{i-1} + \sum\limits_{k=i+1}^{n-1} \frac{\lambda_k}{k-1}.
         \label{Eq2}
     \end{align}
     Finally, for $i+1 \leq j \leq n$, we have
     \begin{align}
          FP_{\widetilde{T}}(x_j) = \sum\limits_{k=1}^{j-1} \lambda_k + \sum\limits_{k=j}^{n-1} \frac{\lambda_k}{k-1}.
          \label{Eq3}
     \end{align}
     Comparing Equations \eqref{Eq1}, \eqref{Eq2}, and \eqref{Eq3}, we clearly have
     \begin{align*}
         FP_{\widetilde{T}}(x_1) \leq \ldots \leq FP_{\widetilde{T}}(x_{i-1}) \leq FP_{\widetilde{T}}(x_{i+1}) \leq \ldots \leq FP_{\widetilde{T}}(x_n).
     \end{align*}
     In particular, if $FP_T(x_i) \geq FP_T(x_j)$, then we also have $FP_{\widetilde{T}}(x_i) \geq FP_{\widetilde{T}}(x_j)$ for all $x_i,x_j \in \widetilde{X}$.
\end{itemize}
Now, if $|X'| > 1$, we sequentially delete all other elements of $X'$ and repeat the argument above. This completes the proof.

\begin{figure}[htbp]
\centering
\includegraphics[scale=0.25]{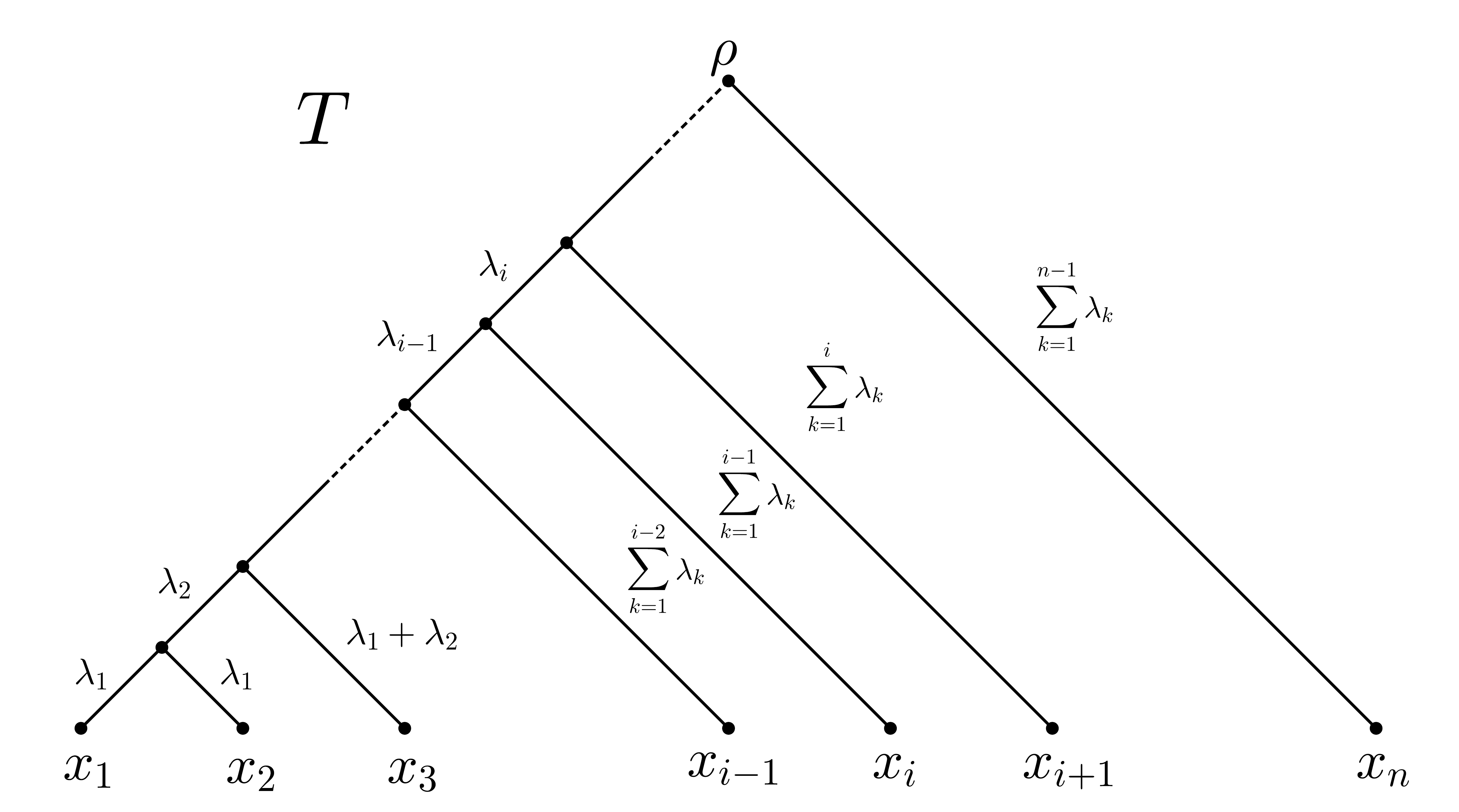}
\caption{Ultrametric caterpillar tree $T$ on $X = \{x_1, \ldots, x_n\}$ used in the proof of Proposition~\ref{Prop_caterpillar_ultrametric}.}
\label{Fig_UltrametricCaterpillar}
\end{figure}
\end{proof}

\end{document}